\documentclass[12pt]{amsart}
\usepackage{amsmath}
\usepackage{amsfonts}
\usepackage{amsthm}
\usepackage{bbm}
\usepackage{mathtools}

\usepackage{amscd,amssymb,latexsym}
\usepackage[mathscr]{euscript}
\usepackage{epsfig}
\usepackage{color}

\allowdisplaybreaks 

\newtheorem{theorem}{Theorem}[section]
\newtheorem{lemma}[theorem]{Lemma}

\newtheorem{example}[theorem]{Example}
\newtheorem{corollary}[theorem]{Corollary}
\newtheorem{remark}[theorem]{Remark}
\newtheorem{Prop}[theorem]{Proposition}

\newcommand{\N}{\mathbb{N}}
\newcommand{\Z}{\mathbb{Z}}
\newcommand{\R}{\mathbb{R}}

\newcommand{\C}{\mathbb{C}}

\newcommand{\1}{\mathbbm 1}

\newcommand{\ep}{\epsilon}

\newcommand{\bs}{\backslash}
\newcommand{\ra}{\rightarrow}

\def\tr{\operatorname{tr}}
\def\Ran{\operatorname{Ran}}
\newcommand{\ten}{\otimes}

\newcommand{\mc}{\mathcal}
\newcommand{\inv}{^{-1}}
\newcommand{\bm}{\boldmath}

\newcommand{\pa}{\mathcal{P}_{\textit{ar}}}

\newcommand{\X}{\mathbf{X}}
\newcommand{\Y}{\mathbf{Y}}

\newcommand{\cdott}{\displaystyle\raisebox{-7.2pt}{\scalebox{3.8}{\ensuremath{\cdot}}}}

\begin{document}

\title{On the nonlinearity of quantum dynamical entropy}
\author{George Androulakis} 
\address{Department of Mathematics, University of South Carolina,
	Columbia, SC 29208 U.S.A.}
\email{giorgis@math.sc.edu}	
\email{dw7@math.sc.edu}

\author{Duncan Wright}
\thanks{This work is part of the second author's PhD thesis.}
\thanks{The second author was partially supported by a SPARC Graduate Research Grant from the Office of the Vice President for Research at the University of South Carolina.} 
\maketitle

\begin{abstract}
Linearity of a dynamical entropy means that the dynamical entropy of the $n$-fold composition of a dynamical map with itself is equal to $n$ times the dynamical entropy of the map for every positive integer $n$. 
We show that the quantum dynamical entropy introduced by S\l{}omczy\'nski and \.Zyczkowski is nonlinear in the time interval between successive measurements of a quantum dynamical system. This is in contrast to Kolmogorov-Sinai dynamical entropy for classical dynamical systems, which is linear in time. We also compute the exact values of quantum dynamical entropy for the Hadamard walk with varying L\"uders-von Neumann instruments and partitions. 
\end{abstract}


\section{Introduction}

Entropy is a crucial concept in thermodynamics, dynamical systems and information theory. 
It was first introduced mathematically by Boltzmann near the end of the 19th century as a tool to measure disorder for the positions and velocities of gas molecules \cite{Boltzmann66}. 
Almost eighty years later, Shannon became the father of the new field of information theory when he produced his groundbreaking works where he used entropy as a measure of information transfer between two sources \cite{Shannon48a, Shannon48b}. 
Dynamical entropy in classical systems can be seen from two distinct viewpoints: The information theoretic viewpoint (see e.g. \cite{CT91}) which uses entropy rate of stochastic processes and the dynamical systems viewpoint (see e.g. \cite{Downarowicz11}) which uses the Kolmogorov-Sinai (KS) dynamical entropy. 
We prove that the connection between entropy rate and KS entropy is seen through the symbolic dynamics of a stochastic process, which is a dynamical system with KS entropy equal to the entropy rate of the original stochastic process (see Subsection~\ref{symbolic dynamics}). 
On the other hand, entropy rate and KS entropy are inherently different as the former is probabilistic in nature and the latter is deterministic (see Subsection~\ref{differences}). 

There have been many successful attempts to generalize KS entropy to a quantum dynamical entropy in \cite{CNT87, AF94, SZ94, AOW97, KOW99} and more. 
The Connes-Narnhofer-Thirring (CNT) \cite{CNT87}, Alicki-Lindblad-Fannes (ALF) \cite{AF94}, Accardi-Ohya-Watanabe (AOW) \cite{AOW97} and Kossakowski-Ohya-Watanabe (KOW) \cite{KOW99} entropies have had the most attention in the literature as they can be computed exactly for several examples of quantum dynamical systems. 
However we will investigate the S\l{}omczy\'nski-\.Zyczkowski (SZ) quantum dynamical entropy \cite{SZ94} which uses a semi-classical approach and was developed using the general notions of measurements, instruments, phase space and state space developed by Edwards \cite{Edwards70} and Davies and Lewis \cite{DL70, Davies76}.   
In contrast to both the CNT and ALF entropies, SZ dynamical entropy can obtain nonzero values for quantum systems with finite-dimensional Hilbert spaces. 
Quantum algorithms are a natural example where this property is desirable. 
To demonstrate the applicability of SZ dynamical entropy for quantum algorithms, we will introduce unitary quantum random walks which have been shown to be universal for quantum computation \cite{LCET10} and give exact computations of the SZ dynamical entropy for the Hadamard walk measured with L\"uders-von Neumann instruments. 

The paper is organized as follows: 
In Section~\ref{classical systems} we recall the definitions of entropy rate and KS entropy. 
In Subsection~\ref{symbolic dynamics} we introduce symbolic dynamics and establish the connection between the two notions of entropy and in Subsection~\ref{differences} we show how they two differ. 
In Section~\ref{measurements} we recall the general notions of measurements through state space, phase space, observables and instruments. 
In Section~\ref{Dyn Ent of UQRWs} we recall the notions of SZ dynamical entropy and show that it is nonlinear in time intervals between successive measurements, in contrast to KS entropy for classical dynamical systems. 
In Section~\ref{Hadamard walk section} we recall unitary quantum random walks and, in particular, the Hadamard walk. 
Lastly, in Section~\ref{UQRW dynamical entropy} we give exact calculations of the SZ dynamical entropy of the Hadamard walk with varying instruments, partitions and time intervals between successive measurements. These calculations verify the nonlinearity of the SZ dynamical entropy via an explicit example (see Theorem~\ref{coarser instrument}). 


\section{Entropy in Classical Systems}\label{classical systems}

\subsection{Entropy from the Dynamical Systems Point of View}\label{dynamical systems}

Let $(\Omega, \Sigma)$ be a measurable space. 
Define $\pa(\Omega)$ to be the collection of all finite or countably infinite measurable partitions of $\Omega$. 
Define a partial ordering on $\pa(\Omega)$ such that, for any $\mc C,\mc D\in \pa(\Omega)$, $\mc D\le \mc C$ whenever, for every $D\in \mc D$ there exists $\mc C_D\subseteq \mc C$ such that $D=\cup \mc C_D$. 
If $\mc D\le \mc C$ we say that \textbf{\bm{$\mc C$} is finer than \bm{$\mc D$}} or that \textbf{\bm{$\mc D$} is coarser than \bm{$\mc C$.}} 

Whenever $\Omega$ is finite or countably infinite and $\Sigma=\mc P(\Omega)$, where $\mc P(\Omega)$ is the power set of $\Omega$, we call $(\Omega,\Sigma)$ (or simply $\Omega$) a \textbf{discrete space}. 
In that case, we will refer to the partition of $\Omega$ into singletons  $\{\{\omega\}\}_{\omega\in \Omega}$ as the \textbf{atomic partition}. 
Whenever $\Omega$ is a discrete space it is clear that the atomic partition $\mc A$ of $\Omega$ is countable and measurable; i.e. $\mc A\in\pa(\Omega)$. 
Furthermore, in this case, $\mc A$ is the finest partition in $\pa(\Omega)$; i.e. $\mc C\le \mc A$ for any $\mc C\in \pa(\Omega)$.  

For any $\mc C, \mc D\in \pa(\Omega)$, the join (or least upper bound) of $\mc C$ and $\mc D$ is given by the partition $\mc C\vee \mc D$ which contains all sets of the form $C\cap D$ for all $C\in\mc C$ and $D\in \mc D$.
Given a finite collection of partitions $\{\mc C_k\}_{k=1}^n\subseteq \pa(\Omega)$, we define the join $\vee_{k=1}^n \mc C_k$ to be the partition containing exactly the sets of the form $\cap_{k=1}^n C_{k}$, where $C_{k}\in \mc C_k$ for all $1\le k\le n$. 

Fix a probability space $(\Omega,\Sigma, \mu)$. 
Given any partition $\mc C\in \pa(\Omega)$ we define the \textbf{entropy of \bm{$\mc C$}} by 
\begin{equation*}
H_\mu(\mc C):=\sum_{C\in \mc C} \eta(\mu(C)),
\end{equation*}

\noindent where $\eta:[0,\infty)\ra [0,\infty)$ is given by $\eta(x)=-x\ln x$, for $x>0$ and we agree that $\eta(0)=0$. 
When there is no confusion about the probability measure in question, we will simply write $H(\mc C)$ instead of $H_\mu(\mc C)$. 

\begin{remark}[\text{\cite[Page 23]{Downarowicz11}}]\label{subadditivity}
	It is well known that $\eta$ is countably subadditive; i.e. $\eta(\sum_n a_n)\le \sum_n \eta(a_n)$ for any nonnegative sequence $\{a_n\}_n$.
	This gives that for any probability space $(\Omega,\Sigma,\mu)$ and any two partitions $\mc C,\mc D\in \pa(\Omega)$ satisfying $\mc D\le \mc C$, we have that $H(\mc D)\le H(\mc C)$. 
\end{remark} 

Fix a probability space $(\Omega,\Sigma,\mu)$. 
Recall that, for any two sets $C,D\in \Sigma$, the \textbf{conditional probability of \bm{$C$} given \bm{$D$}} is given by $\mu(C|D):=\mu(C\cap D)/\mu(D)$. 
Given two partitions, $\mc C,\mc D\in\pa(\Omega)$, the \textbf{conditional entropy of \bm{$\mc C$} given \bm{$\mc D$}} is given by 
\begin{equation}\label{conditional entropy 3}
H(\mc C|\mc D):= \sum_{D\in\mc D} \mu(D)\sum_{C\in\mc C} \eta(\mu(C|D)) =-\sum_{\substack{C\in\mc C\\ D\in\mc D}} \mu(C\cap D)\ln(\mu(C|D)).
\end{equation} 

\noindent The so-called chain rule follows. 

\begin{theorem}[Chain Rule, \text{\cite[Equation 1.4.3]{Downarowicz11}}]\label{chain rule} 
Let $(\Omega,\Sigma,\mu)$ be a probability space and $\mc C,\mc D\in\pa(\Omega)$. 
Then
\begin{equation*}
H(\mc C\vee \mc D)=H(\mc D) + H(\mc C|\mc D). 
\end{equation*}

\noindent More generally, given a finite collection of partitions $\mc C_0,\ldots, \mc C_n\in\pa(\Omega)$, we have 
\begin{equation*}
H(\vee_{k=0}^n \mc C_k)=H(\mc C_0) + \sum_{k=1}^{n} H(\mc C_k|\vee_{\ell=0}^{k-1}\mc C_\ell). 
\end{equation*}
\end{theorem}

\noindent Also, from the definition of conditional entropy  (Equation~\eqref{conditional entropy 3}) 
and the countable subadditivity of $\eta$ 
(Remark~\ref{subadditivity})
we have, for any $\mc B, \mc C,\mc D\in \pa(\Omega)$ satisfying $\mc B\le \mc D$, that 
\begin{equation}\label{conditional subadditivity}
0\le H(\mc C|\mc D)\le H(\mc C|\mc B). 
\end{equation}

\noindent See \cite[Section 1.4]{Downarowicz11} for more details on conditional entropy of partitions. 
The following theorem will be used throughout the article. 
In the proof we will use the well known C\'esaro mean Theorem which states that, for any sequence of real numbers $\{a_n\}_{n=0}^\infty$ converging to some element, $a\in\R\cup \{\infty\}$, the sequence $\{b_n\}_{n=0}^\infty$ given by $b_n=\frac{1}{n}\sum_{k=0}^{n-1} a_k$, for each $n\in\N$, also converges to $a$. 

\begin{theorem}\label{entropy equivalence}
Let $(\Omega,\Sigma,\mu)$ be a probability space and $\{\mc C_n\}_{n=0}^\infty\subseteq \pa(\Omega)$ be a sequence of partitions. 
If $\lim_{n\ra\infty} H(\mc C_n|\vee_{k=0}^{n-1} \mc C_k)=a$, then 
$\lim_{n\ra\infty}\frac{1}{n} H(\vee_{k=0}^{n-1}\mc C_k)=a$. 
\end{theorem}

\begin{proof}
Set $a_0=H(\mc C_0)$ and $a_n= H(\mc C_n|\vee_{k=0}^{n-1} \mc C_k)$, for each $n\in\N$. 
Then, by assumption, $a_n$ converges to $a\in\R\cup \{\infty\}$. 
For each $n\in\N$, set $b_n=\frac{1}{n}\sum_{k=0}^{n-1}a_k$. 
Then, by Theorem~\ref{chain rule}, $b_n=\frac{1}{n}H(\vee_{k=0}^{n-1}\mc C_k)$ which converges to $a$ by the C\'esaro mean Theorem. 
\end{proof}

Next we wish to define the Kolmogorov-Sinai (KS) dynamical entropy. 
In order to do so, we must first introduce some dynamics on the probability space $(\Omega,\Sigma,\mu)$.  
This role will be played by a measurable map $f:\Omega\ra\Omega$. 
We call the quadruple $(\Omega,\Sigma,\mu,f)$ a \textbf{dynamical system} (DS). 
Furthermore, whenever $\mu(A)=\mu(f\inv(A))$ for all $A\in \Sigma$, we say that \textbf{\bm{$\mu$} is \bm{$f$}-invariant} and call the DS $(\Omega,\Sigma,\mu,f)$ \textbf{stationary.} 

Fix a DS $(\Omega,\Sigma,\mu,f)$ and a partition $\mc C\in \pa(\Omega)$. 
For each $k\in \N_0$, define the $k^\text{th}$-preimage of $\mc C$ under $f$ by $f^{-k}(\mc C):=\{f^{-k}(C)\}_{C\in \mc C}$, where $f^0$ denotes the identity map. 
Note that, for each $\mc C\in \pa(\Omega)$, $f\inv (\mc C)\in \pa(\Omega)$ and hence $f^{-k}(\mc C)\in \pa(\Omega)$ for every $k\in\N_0$.
The \textbf{Kolmogorov-Sinai (KS) entropy of \bm{$(\Omega,\Sigma,\mu,f)$} with respect to \bm{$\mc C$}} is given by 
\begin{equation}\label{invariant state limits}
h^{KS}(f,\mc C)=\lim_{n\ra\infty} \frac{1}{n}H(\vee_{k=0}^{n-1} f^{-k}(\mc C)),
\end{equation}

\noindent whenever this limit exists.

\begin{remark}\label{pullbacks}
	Given a DS $(\Omega,\Sigma,\mu,f)$ and a partition $\mc C\in\pa(\Omega)$, it is clear that $\vee_{k=0}^{n-1} f^{-k}(\mc C)$ consists exactly of sets of the form $f^{-(n-1)}(A_{n-1})\cap \cdots \cap f\inv (A_1)\cap A_0$ for all $A_0, \ldots, A_{n-1}\in\mc C$. 
\end{remark} 

From Theorem~\ref{entropy equivalence}, we have   
\begin{equation}\label{KS entropy 2}
h^{KS}(f,\mc C)=\lim_{n\ra\infty} H(f^{-n}(\mc C)|\vee_{k=0}^{n-1} f^{-k}(\mc C)),
\end{equation}

\noindent whenever this limit exists. 

\begin{corollary}\label{stationary DS}
	Let $(\Omega,\Sigma,\mu,f)$ be a stationary DS and $\mc C\in\pa(\Omega)$. 
	Then the limit in Equation~\eqref{KS entropy 2}, and hence the limit in Equation~\eqref{invariant state limits}, exists and $h^{KS}(f,\mc C)=\lim_{n\ra\infty} H(f^{-n}(\mc C)|\vee_{k=0}^{n-1} f^{-k}(\mc C))$.
\end{corollary}

\begin{proof}
	For each $n\in\N$ with $n\ge 2$, we have 
	\begin{align*}
	H(f^{-n}(\mc C)|\vee_{k=0}^{n-1} f^{-k}(\mc C)) &\le  
	H(f^{-n}(\mc C)|\vee_{k=1}^{n-1} f^{-k}(\mc C)) 
	\quad\text{by \eqref{conditional subadditivity}} \\
	&= H(f^{-(n-1)}(\mc C)|\vee_{k=0}^{n-2} f^{-k}(\mc C))
	\quad\text{since $(\Omega,\Sigma,\mu,f)$ is stationary.}
	\end{align*}
	
	\noindent Therefore $H(f^{-n}(\mc C)|\vee_{k=0}^{n-1} f^{-k}(\mc C))$ is a decreasing sequence which is bounded below by zero and hence converges. 
	By Theorem~\ref{entropy equivalence}, $$h^{KS}(f,\mc C)=\lim_{n\ra\infty} H(f^{-n}(\mc C)|\vee_{k=0}^{n-1} f^{-k}(\mc C)).$$
\end{proof}

\begin{remark}
	In the literature it is common to only refer to $(\Omega,\Sigma,\mu,f)$ as a DS whenever $\mu$ is $f$-invariant. 
	Although this convention has its benefits, as evidenced by Corollary~\ref{stationary DS}, we find it restrictive and do not adopt it here.  
\end{remark}

Finally, the \textbf{KS entropy of \bm{$(\Omega,\Sigma,\mu,f)$}} is given by 
\begin{equation}\label{KS entropy}
h^{KS}(f)=\sup_{\substack{\mc C\in \pa(\Omega)\\ H(\mc C)<\infty}}h^{KS}(f,\mc C).
\end{equation}

\begin{remark}\label{only finite partitions}
	Fix a dynamical system $(\Omega,\Sigma,\mu, f)$. 
	In many instances, KS entropy is taken as the $\sup$ over only finite partitions. 
	However, the two definitions are equivalent (see \cite[Page 102]{Downarowicz11}).
	Furthermore, it is remarked in \cite[Page 61]{Downarowicz11} that the restriction of the $\sup$ in Equation~\eqref{KS entropy} to include only those partitions, $\mc C$, satisfying $H(\mc C)<\infty$ is natural because otherwise it is possible to obtain infinite KS entropy for the identity transformation. 
	This is due to the fact that $H(f,\mc C)=\infty$ whenever $H(\mc C)=\infty$.
\end{remark}

For a more detailed exposition on dynamical entropy and classical dynamical systems (with invariant measures), we refer the reader to the book of Walters \cite{Walters82}. 
For extensions of the results of Walters to include infinite partitions with finite entropy, the reader is referred to the book of Downarowicz \cite{Downarowicz11}.

\subsection{Entropy from the Information Theoretic Point of View}\label{information theoretic}

Next we look at entropy of random variables. 
Let $(\Omega,\Sigma,\mu)$ be a probability space  and let $(E,\mc E)$ be a measurable space. 
An \textbf{\bm{$(\Omega,E)$} random variable} $X$ is a measurable map $X:\Omega\ra E$. 
For any $S\in\mc E$, the probability that $X$ takes values in $S$ is given by $\mu(X\in S):=\mu(X\inv (S))$. 
We say that the random variable \textbf{\bm{$X$} is discrete} whenever its range $E$ is a discrete space. 
In that case, $X$ is determined by its \textbf{probability mass function} (pmf) $p_X:E\ra [0,1]$ given by $p_X(x)=\mu(X=x)$ for each $x\in E$. 
We will simply write $p(x)$ as opposed to $p_X(x)$ when there is no confusion about the random variable. 
Given a discrete random variable $X$, the collection 
\begin{equation*}
\mc C_X:=\{X\inv (\{x\})\}_{x\in E}
\end{equation*}

\noindent is a countable measurable partition of $\Omega$; i.e. $\mc C_X\in\pa(\Omega)$. 
The \textbf{entropy} of a discrete random variable $X$ is given by the entropy of $\mc C_{X}$ and is related to its pmf by the equation 
\begin{equation}\label{static entropy}
H_\mu(X):=H_\mu(\mc C_X)=\sum_{x\in E} \eta(p(x)).
\end{equation}

\noindent When there is no confusion about the probability measure in question, we will simply write $H(X)$ instead of $H_\mu(X)$. 

Given a finite collection, $(X_k)_{k=0}^n$, of $(\Omega, E)$ discrete random variables  the \textbf{joint pmf of \bm{$(X_0,\ldots, X_n)$}} is given by $p_{X_0,\ldots, X_n}(x_0,\ldots, x_n)=\mu(X_0=x_0, \ldots, X_n=x_n)$ for all $x_0,\ldots, x_n\in E$. 
Furthermore, $(X_0,\ldots, X_n)$ is a discrete $(\Omega,E^{n+1})$ random variable and $\mc C_{(X_0,\ldots, X_n)}=\vee_{k=0}^n \mc C_{X_k}$. 
Therefore the entropy of $(X_0,\ldots, X_n)$ is given by Equation~\eqref{static entropy} and is related to its joint pmf by the equation 
\begin{equation}\label{static entropy of process}
H(X_0,\ldots,X_n):= H(\vee_{k=0}^n \mc C_{X_k})=\sum_{\substack{x_k\in E \\ 0\le k\le n}} \eta(p_{X_0,\ldots, X_n}(x_0,\ldots, x_n)). 
\end{equation}

\noindent The \textbf{conditional pmf of \bm{$X_n$} given \bm{$(X_0,\ldots, X_{n-1})$}} is given by 
\begin{equation*}
\begin{split}
p_{X_n|(X_0,\ldots, X_{n-1})}(x_n|x_0,\ldots, x_{n-1}):&= 
\mu(X_n=x_n|X_0=x_0,\ldots, X_{n-1}=x_{n-1}) \\
&= \frac{p_{X_0,\ldots, X_n}(x_0,\ldots, x_n)}{p_{X_0,\ldots, X_{n-1}}(x_0,\ldots, x_{n-1})}
\quad\text{for all }x_0,\ldots, x_n\in E.
\end{split}
\end{equation*}

\noindent Whenever there is no confusion about the random variables in question we simply write $p(x_0,\ldots, x_n)$ for the joint pmf and $p(x_n|x_0,\ldots,x_{n-1})$ for the conditional pmf. 
The \textbf{conditional entropy of \bm{$X_n$} given \bm{$(X_0,\ldots, X_{n-1})$}} is given by the conditional entropy of $C_{X_n}$ given $\vee_{k=0}^{n-1} C_{X_k}$ and is related to their conditional pmf by the equation 
\begin{equation}\label{conditional entropy}
\begin{split}
H(X_n|X_0,\ldots, X_{n-1}):&= H(\mc C_{X_n}|\vee_{k=0}^{n-1}\mc C_{X_k}) \\
&= \sum_{\substack{x_k\in E \\ 0\le k\le n-1}} p(x_0,\ldots, x_{n-1}) \sum_{x_n\in E}\eta(p(x_n|x_0,\ldots, x_{n-1})).
\end{split} 
\end{equation}

\noindent See \cite[Section 2.2]{CT91} 
for more details on conditional entropy of random variables. 

\begin{remark}
The entropy and conditional entropy of non-discrete random variables can be defined similarly to Equations~\eqref{static entropy} and \eqref{conditional entropy}, respectively, by using integration and probability distribution functions instead of sums and pmfs. 
However, we are mainly interested in discrete random variables here. 
\end{remark}

Next we turn to entropy rate of stochastic processes. 
If $(\Omega,\Sigma,\mu)$ is a probability space and $(E,\mc E)$ is a measurable space, then a \textbf{(discrete time) \bm{$(\Omega, E)$} stochastic process} is an indexed sequence of $(\Omega, E)$ random variables. 
Throughout this paper we will only consider discrete time stochastic processes and the sequences will all be indexed by $\N_0:=\N\cup\{0\}$, where the index is meant to represent time. 
Whenever $(E,\mc E)$ is a discrete space the stochastic process, $\X=(X_n)_{n=0}^\infty$, is determined by its joint pmf, denoted by $p_{\X}$, and given by $p_{\X}(x_0,\ldots, x_n)=p_{X_0,\ldots, X_n}(x_0,\ldots, x_n)$ for each $n\in\N_0$ and $x_0,\ldots, x_n\in E$. 
Given a stochastic process of discrete random variables, the entropy of a finite initial subsequence is given by Equation~\eqref{static entropy of process} and the conditional entropy of the $n^{\text{th}}$ term given all the previous ones is given by Equation~\eqref{conditional entropy}. 

Fix a discrete $(\Omega,E)$ stochastic process $\X$. 
Then, from the classical random walk perspective, we interpret $p_{X_0}(x)$ as the probability that a random walker inhabits the site $x$ initially at time 0 and $p_{X_n}(x)$ as the probability that a random walker inhabits the site $x$ at time $n$, for any $x\in E$ and $n\in \N$. 
Furthermore, for any $x_0,\ldots, x_n\in E$ and $n\in\N$, we interpret $p_(x_0,\ldots, x_n)$ as the probability that a random walker takes the path $x_0\ra x_1\ra \cdots \ra x_n$ at times $0,1,\ldots, n$. 

A stochastic process $(X_n)_{n=0}^\infty$ of discrete random variables is called \textbf{stationary} whenever its joint pmf is invariant with respect to shifts of the time index; i.e.
\begin{equation*}
\mu(X_0=x_0,\ldots, X_n=x_n)=\mu(X_l=x_0,\ldots, X_{n+l}=x_n), 
\end{equation*}

\noindent for all $n,l\in\N_0$ and $x_0,\ldots,x_{n}\in E$.
In the literature, a stationary stochastic process is sometimes referred to as being time invariant (see e.g. \cite[Page 61]{CT91}).

A simple example of a stochastic process is one in which each random variable depends only on the one proceeding it in the sequence; i.e. 
\begin{equation}\label{Markov process}
\mu(X_{n+1}=x_{n+1}|X_0=x_0,\ldots, X_n=x_n) = \mu(X_{n+1}=x_{n+1}|X_n=x_n),
\end{equation}

\noindent for all $n\in\N_0$ and $x_0,\ldots,x_{n+1}\in E$. 
A stochastic process of discrete random variables which satisfies Equation~\eqref{Markov process} is referred to as a \textbf{Markov process.} 
In particular, we are interested in those discrete Markov processes, $\X=(X_n)_{n=0}^\infty$, whose conditional pmfs do not vary with time; i.e. 
\begin{equation}\label{Markov process 2}
\mu(X_1=x|X_0=y)=\mu(X_{n+1}=x|X_n=y)\text{ for all $x,y\in E$, and $n\in\N_0$.}
\end{equation}

\noindent In this case, we will set $p_{x, y}:= \mu(X_1=x|X_0=y)$ and define the $|E|\times |E|$ matrix $P$ 
to have $(x,y)$-entry given by $p_{x, y}$, for all $x,y\in E$. 
Then $P$ is a \textbf{transition (column-stochastic) matrix}; i.e. for all $x,y\in E$, $0\le p_{x, y}\le 1$ and, for all $y\in E$, $\sum_{x\in E}p_{x, y}=1$. 
From the classical random walk perspective, the $(x,y)$-entry of $P$, $p_{x, y}$, is interpreted as the conditional probability that a random walker will move in one step from site $y$ to site $x$. 

The \textbf{entropy rate of a stochastic process \bm{$\X=(X_n)_{n=0}^\infty$}} is given by 
\begin{equation}\label{entropy rate}
H(\X):=\lim_{n\ra\infty}\frac{1}{n} H(X_0,\ldots, X_{n-1}),
\end{equation}

\noindent whenever this limit exists. 

Another quantity, which is often equal to the entropy rate, is given by 
\begin{equation}\label{entropy rate 2}
H'(\X):=\lim_{n\ra\infty} H(X_{n}|X_0,\ldots, X_{n-1}),
\end{equation}

\noindent whenever this limit exists. 
The two quantities $H(\X)$ and $H'(\X)$ correspond to two different interpretations of entropy rate. 
The first is interpreted as the average entropy of the first $n$ random variables and the second as the entropy of the last random variable given the past. 
The following result shows the relationship between $H(\X)$ and $H'(\X)$. 

\begin{corollary}\label{two entropy rates}
	Let $\X=(X_n)_{n=0}^\infty$ be a stochastic process. 
	If the limit in Equation~\eqref{entropy rate 2} exists, 
	then the limit in Equation~\eqref{entropy rate} also exists and $H(\X)=H'(\X)$.
\end{corollary}

\begin{proof}
By the definitions of $H(X_0,\ldots,X_n)$ and $H(X_n|X_0,\ldots,X_{n-1})$ in Equations~\eqref{static entropy of process} and \eqref{conditional entropy}, respectively, this is simply a restatement of Theorem~\ref{entropy equivalence}. 
\end{proof}

%

The following is another corollary for stationary stochastic processes which is also proved in \cite[Theorem 4.2.2]{CT91}. 

\begin{corollary}\label{stationary process}
Let $\X$ be a stationary stochastic process. 
Then the limits in Equations~\eqref{entropy rate} and \eqref{entropy rate 2} both exist and $H(\X)=H'(\X)$. 
\end{corollary} 

\begin{proof}
The proof is similar to the proof of Corollary~\ref{stationary DS}. 
For each $n\in\N$ with $n\ge 2$, we have 
\begin{align*}
H(X_n|X_0,\ldots, X_{n-1}) &\le  
H(X_n|X_1,\ldots, X_{n-1}) 
\quad\text{by \eqref{conditional subadditivity}} \\
&= H(X_{n-1}|X_0,\ldots, X_{n-2})
\quad\text{since $\X$ is stationary.}
\end{align*}

\noindent Therefore $H(X_n|X_0,\ldots, X_{n-1})$ is a decreasing sequence which is bounded below by zero and hence converges. 
By Equations~\eqref{static entropy of process} and \eqref{conditional entropy}, and Theorem~\ref{entropy equivalence}, $H(\X)=H'(\X)$.
\end{proof}

Next we look at the entropy rate of discrete Markov processes governed by a transition matrix; i.e. the discrete Markov processes satisfying Equation~\eqref{Markov process 2}.  
In this case, given a discrete measurable space $(E,\mc P(E))$, we will represent a probability measure $\mu$ on $(E,\mc P(E))$ as the column vector $\mu=\{\mu_e\}_{e\in E}$, where $\mu_e:=\mu(\{e\})$ for each $e\in E$, which we will refer as a \textbf{probability vector}. 
Then, given a transition matrix $P$ on $E$, we define $P\mu$ by matrix multiplication. 
We say that \textbf{\bm{$\mu$} is \bm{$P$}-invariant} whenever $P\mu=\mu$. 
In particular, whenever $\X$ is an $(\Omega,E)$ Markov process governed by the transition matrix $P$, we take the initial measure $\mu$ to be $p_{X_0}$. 
In this case, notice that $\X$ is stationary if and only if $p_{X_0}$ is $P$-invariant. 
The following theorem gives a simplification of the entropy rate for Markov processes governed by a transition matrix.  

\begin{theorem}\label{markov entropy rate}
Let $\X$ be a discrete $(\Omega,E)$ Markov process governed by the transition matrix $P$ 
and set $\mu=p_{X_0}$. 
Then 
\begin{equation*}
H(\X)=\lim_{n\ra\infty} \sum_{y\in E} (P^n\mu)_y\sum_{x\in E}\eta(p_{x, y}),
\end{equation*}
	
\noindent whenever the limit exists. 
Moreover, if $\X$ is stationary, then $$H(\X)=\sum_{y\in E}\mu_y\sum_{x\in E}\eta(p_{x, y}).$$ 
\end{theorem}

\begin{proof}
	Since $\X$ is a Markov process governed by the transition matrix $P$ we have 
	\begin{equation*}
	H(X_{n+1}|X_0,\ldots,X_{n})=H(X_{n+1}|X_{n})=
	\sum_{y\in E}p_{X_{n}}(y)
	\sum_{x\in E}\eta(p_{x, y})
	\quad\text{for each $n\in\N$,} 
	\end{equation*}
	
	\noindent where the second equality follows from Equation~\eqref{conditional entropy}. 
	Then from the definition of matrix multiplication, for each $e_n\in E$, we have  
	\begin{equation*}
	p_{X_n}(e_n)=\sum_{\substack{e_k\in E \\ 0\le k\le n-1}} 
	p_{X_0}(e_0)\prod_{k=1}^n p_{e_k, e_{k-1}}=(P^n\mu)_{e_n}.
	\end{equation*}
	
	\noindent Therefore 
	$$
	H(\X)=\lim_{n\ra\infty} \sum_{y\in E} (P^n\mu)_y\sum_{x\in E}\eta(p_{x, y}),
	$$
	
	\noindent whenever the limit exists. 
	The moreover statement is immediate because $P^n\mu=\mu$ for all $n\in\N_0$ whenever $\X$ is stationary.  
\end{proof}

\begin{remark}
	Certain cases of Theorem~\ref{markov entropy rate} appear frequently in the literature, but to the best of our knowledge we have not seen it presented in the generality of above. 
	For instance, it can be seen for the case where $\mu$ is $P$-invariant in \cite[Theorem 4.2.4]{CT91} or \cite[Theorem 4.26]{Walters82}. 
\end{remark}

In the literature, given a transition matrix $P$ on a discrete measurable space $(E,\mc P(E))$ with a unique invariant probability vector $\mu$, it is common to set 
\begin{equation}
\label{entropy of P}
H(P):= \sum_{y\in E}\mu_y\sum_{x\in E}\eta(p_{x, y}),
\end{equation}

\noindent and refer to $H(P)$ as the \textbf{entropy of \bm{$P$}.} 
As it is shown in Theorem~\ref{markov entropy rate}, the entropy, $H(P)$, of $P$ is equal to the entropy rate, $H(\X)$, of any stationary Markov process, $\X$, governed by the transition matrix $P$ such that $p_{X_0}=\mu$.


\subsection{The connection between entropy rate and KS entropy}\label{symbolic dynamics}

Let $(\Omega,\Sigma,\mu)$ be a probability space, $(E,\mc E)$ a measurable space and $\X=(X_n)_{n=0}^\infty$ an $(\Omega,E)$ stochastic process. 
Consider the measurable space $(E^*,\mc E^*)$, where $E^*:=E^{\N_0}$ and $\mc E^*:=\sigma(\cup_{n=0}^\infty \mc E^{n+1})$. 
For all $n\in\N_0$, collection of integer times $0\le t_0<t_1<\cdots<t_n$ and $A_0,\ldots, A_n\in \mc E$, we define the \textbf{cylinder set} 
\begin{equation*}
C\left(\begin{smallmatrix} A_0 & \cdots & A_{n} \\ t_0 & \cdots & t_n\end{smallmatrix}\right) := \{ x=(x_i)_{i\in\N_0}\in E^* : x_{t_k}\in A_k\text{ for } k\in \{0,\ldots, n\}\}.
\end{equation*} 

\noindent For $\mc C\in\pa(E)$, we say $C\left(\begin{smallmatrix} A_0 & \cdots & A_{n} \\ t_0 & \cdots & t_n\end{smallmatrix}\right)$ is a \textbf{\bm{$\mc C$}-cylinder set} if $A_0,\ldots, A_n\in\mc C$.
Also, we define the partition, $\widehat{\mc C}\in \pa(E^*)$, by $$\widehat{\mc C}:=\{C\left(\begin{smallmatrix} A \\ 0\end{smallmatrix}\right)\}_{A\in \mc C}$$ and the set  $$\widehat{\pa(E)}:=\{\widehat{\mc C}: \mc C\in\pa(E)\}\subset \pa(E^{*}).$$ 

Notice that the collection of all cylinder sets in $E^*$ is a $\pi$-system which generates the $\sigma$-algebra $\mc E^*$. 
Therefore, any measure on $(E^*,\mc E^*)$ is uniquely defined by its values on the cylinder sets. 
We define the process-dependent measure, $\mu^{\X}$, on the cylinder sets by   
\begin{equation}\label{symbolic dynamics measure}
\mu^{\X} (C\left(\begin{smallmatrix} A_0 & \cdots & A_{n} \\ t_0 & \cdots & t_n\end{smallmatrix}\right))=\mu(\cap_{k=0}^n (X_{t_k}\in A_k)),
\end{equation}

\noindent for all $n\in\N_0$, $0\le t_0<\cdots<t_n$, and $A_0,\ldots, A_n\in \mc E$. 
We interpret $\mu^{\X} (C\left(\begin{smallmatrix} A_0 & \cdots & A_{n} \\ t_0 & \cdots & t_n\end{smallmatrix}\right))$ as the probability that a random walker, governed by the stochastic process $\X$, is in the set $A_0$ at time $t_0$ and takes the path $A_0\ra A_1\ra \cdots\ra A_n$ at times $t_0,\ldots,t_n$. 

\begin{remark}
Fix an $(\Omega, E)$ stochastic process $\X$, $A_0,\ldots, A_n\in \mc E$, integer times $0\le t_0<\cdots < t_n$ and $t\in \N\backslash \{t_0,\ldots, t_n\}$ such that $t_i<t<t_{i+1}$ for some $0\le i < n$. 
Then, since $\mu(\cap_{k=0}^n (X_{t_k}\in A_k))=\mu(\cap_{k=0}^n (X_{t_k}\in A_k)\cap (X_t\in E))$, we have that 
\begin{equation*}
\mu^{\X} (C\left(\begin{smallmatrix} A_0 & \cdots & A_{n} \\ t_0 & \cdots & t_n\end{smallmatrix}\right))=
\mu^{\X} (C\left(\begin{smallmatrix} A_0 & \cdots& A_i & E & A_{i+1}& \cdots & A_{n} \\ t_0&\cdots & t_i&t&t_{i+1}&\cdots& t_n\end{smallmatrix}\right)). 
\end{equation*}

\noindent Similarly, if $t<t_0$ or $t>t_n$. 
Therefore the measure $\mu^{\X}$ is well defined. 
\end{remark}

We define the shift map $s: E^*\ra E^*$ by $s(x)=y$ where $y_i=x_{i+1}$, for each $i\in\N_0$, and refer to the quadruple, $(E^*, \mc E^*, \mu^{\X}, s)$, as the \textbf{symbolic dynamics of \bm{$\X$}.} 
Notice that $(E^*, \mc E^*, \mu^{\X}, s)$ is a DS and thus its partition dependent and independent KS entropies are given by Equations~\eqref{invariant state limits} and \eqref{KS entropy}, respectively. 

Of particular interest is the KS entropy of $(E^*, \mc E^*, \mu^{\X}, s)$ with respect to the partitions in $\widehat{\pa(E)}$. 
For each $\mc C\in \pa(E)$, define the $(\Omega,\mc C)$ stochastic process $\X_{\mc C}=(X^{\mc C}_n)_{n=0}^\infty$ where, for each $n\in\N_0$ and $\omega\in\Omega$, 
$X^{\mc C}_n(\omega)=A$ whenever $X_n(\omega)\in A$; i.e. $X_n^{\mc C}=i_{\mc C}\circ X_n$, where $i_{\mc C}:E\ra \mc C$ is the natural map that assigns to each $e\in E$ the unique $A\in \mc C$ such that $e\in A$. 
Since the values of $\X_{\mc A}$ are singletons, it is clear that $\X$ can be identified with $\X_{\mc A}$ whenever $\mc A$ is the atomic partition of the discrete space $E$. 
The following proposition shows that the KS entropy of $(E^*, \Sigma^*, \mu^{\X}, s)$ with respect to $\widehat{\mc C}$ and the entropy rate of $\X_{\mc C}$ are equal. 

\begin{Prop}\label{relationship 1}
	Let $(\Omega,\Sigma,\mu)$ be a probability space, $(E,\mc E)$ a measurable space, $\X$ an $(\Omega,E)$ stochastic process and $(E^*, \Sigma^*, \mu^{\X}, s)$ the symbolic dynamics of $\X$. 
	Then for each $\mc C\in\pa(E)$, $H(\X_{\mc C})=h^{KS}(s,\widehat{\mc C})$. 
	In particular, whenever $E$ is a discrete space, $H(\X)=h^{KS}(s,\widehat{\mc A})$, where $\mc A$ is the atomic partition of $E$. 
\end{Prop}

\begin{proof}
Notice that, for each $n\in\N_0$, the collection of all $\mc C$-cylinder sets of initial length $n+1$ is given by $\vee_{k=0}^{n}s^{-k}(\widehat{\mc C})$; i.e. 
\begin{equation}\label{C cylinder sets}
\vee_{k=0}^{n}s^{-k}(\widehat{\mc C})= 
\{C\left(\begin{smallmatrix} A_{0} & \cdots & A_{n} \\ 0 & \cdots & n\end{smallmatrix}\right) 
: A_0,\ldots, A_n\in \mc C \}; 
\end{equation}

\noindent Thus, for each $n\in\N_0$, we see that 
\begin{align*}
H(\vee_{k=0}^{n}s^{-k}(\widehat{\mc C})) &= 
\sum_{\substack{A_k\in \mc C \\ 0\le k\le n}} \eta(\mu^{\X}(C\left(\begin{smallmatrix} A_{0} & \cdots & A_{n} \\ 0 & \cdots & n\end{smallmatrix}\right)))
\quad\text{by \eqref{C cylinder sets}} \\
&= \sum_{\substack{A_k\in \mc C \\ 0\le k\le n}}
\eta(\mu(\cap_{k=0}^n (X_k\in A_{k})))
\quad\text{by \eqref{symbolic dynamics measure}} \\
&=H(X^{\mc C}_0,\ldots, X^{\mc C}_n)
\quad\text{by definition of $\X_{\mc C}$.}
\end{align*}

\noindent Hence, from Equations~\eqref{entropy rate} and \eqref{invariant state limits}, $H(\X_{\mc C})=h^{KS}(s,\widehat{\mc C})$ for each $\mc C\in\pa(E)$. 
Whenever $E$ is a discrete space and $\mc A$ is the atomic partition on $E$, since $\X$ can be identified with $\X_{\mc A}$, it follows that $H(\X)=H(\X_{\mc A})=h^{KS}(s,\widehat{\mc A})$. 
\end{proof}

An important tool for computing the KS entropy of a DS is the Kolmogorov-Sinai Theorem. 
First, given a DS $(\Omega,\Sigma,\mu, f)$ and a partition $\mc C\in\pa(\Omega)$, we say that \textbf{\bm{$\mc C$} is a generating partition for \bm{$(\Omega,\Sigma,\mu, f)$}} if 
\begin{equation*}
\sigma(\cup_{n=0}^\infty \vee_{k=0}^n f^{-k}(\mc C) ) =\Sigma.
\end{equation*}

\noindent Notice that the definition of a generating partition does not depend on $\mu$, but for simplicity of notation we keep the full DS. 

\begin{theorem}[Kolmogorov-Sinai Theorem]\label{KS Theorem} 
Let $(\Omega,\Sigma,\mu, f)$ be a DS and $\mc C,\mc D\in\pa(\Omega)$. 
If $\sigma(\mc D)\subseteq \sigma(\cup_{n=0}^\infty \vee_{k=0}^n f^{-k}(\mc C) )$, then 
$$
h^{KS}(f,\mc C) \ge h^{KS}(f,\mc D).
$$
	
\noindent In particular, if $\mc C$ is a generating partition and $H(\mc C)<\infty$ then $h^{KS}(f)=h^{KS}(f,\mc C)$.
\end{theorem}

A proof of Theorem~\ref{KS Theorem} can be found in \cite[Theorem 4.2.2]{Downarowicz11}.

\begin{corollary}\label{KS theorem 2}
	Let $(\Omega,\Sigma,\mu)$ be a probability space, $(E,\mc E)$ a discrete measurable space, $\mc A$ the atomic partition of $E$, $\X$ an $(\Omega,E)$ stochastic process and $(E^*,\Sigma^*,\mu^{\X},s)$ the symbolic dynamics of $\X$. 
	Then $H(\X)=h^{KS}(s)=h^{KS}(s,\widehat{\mc A})$ whenever $X_0$ has finite entropy.  
\end{corollary} 

\begin{proof}
	From Equation~\eqref{C cylinder sets} we see that $\widehat{\mc A}$ is a generating partition for $(E^*, \Sigma^*, \mu^{\X}, s)$. 
	Then Proposition~\ref{relationship 1} and Theorem~\ref{KS Theorem} give that $h^{KS}(s)=h^{KS}(s,\widehat{\mc A})=H(\X)$, whenever $\widehat{\mc A}$ has finite entropy. 
	Noticing that  $H(\widehat{\mc A})= H(\mc A_{X_0})=H(X_0)<\infty$, the result follows.  
\end{proof}

\begin{remark}
Notice that the results in Proposition~\ref{relationship 1} and Corollary~\ref{KS theorem 2} look nearly identical except that the condition $H(X_0)<\infty$ has been added to the latter. 
This assumption is necessary due to the fact that $h^{KS}(s,\widehat{\mc A})$ is defined regardless of whether $H(\widehat{\mc A})$ is finite or infinite, but is only considered in the supremum of Equation~\eqref{KS entropy} when $H(\widehat{\mc A})$ is finite.
\end{remark}
  
A good resource for symbolic dynamics of Markov processes is \cite{Kitchens98}.


\subsection{Differences between entropy rate and KS entropy}\label{differences}

In this section we give the differences between entropy rate and KS entropy. 
The first thing to notice is that dynamics of a stochastic process are probabilistic in nature, whereas the dynamics of a DS are deterministic in nature. 
This fact will be exploited to establish the differences in the two entropies in this section and again in Section~\ref{Dyn Ent of UQRWs} to establish differences between quantum dynamical entropy and KS entropy. 
The following two propositions give properties of KS entropy whose analogous statements do not hold true for entropy rate. 
The first proposition will use the well known fact (see e.g. \cite[Equation~(1.3.2)]{Downarowicz11}) that for any probability space $(\Omega,\Sigma,\mu)$ and partition $\mc C\in\pa (\Omega)$, we have 
\begin{equation}\label{upper bound}
H(\mc C)\le \ln |\mc C| \le \ln |\Omega|.
\end{equation}

\begin{Prop}\label{finite deterministic system}
Let $(\Omega,\Sigma, \mu, f)$ be a DS such that $|\Omega|<\infty$. 
Then $h^{KS}(f)=0$. 
\end{Prop}

\begin{proof}
For any partition $\mc C\in\pa(\Omega)$, Equation~\eqref{upper bound} gives that $H(\vee_{k=0}^{n-1}f^{-k}(\mc C))\le\ln |\Omega|$, for each $n\in\N$.  
Therefore $h^{KS}(f,\mc C)=0$ for every $\mc C\in\pa(\Omega)$ and thus $h^{KS}(f)=0$.
\end{proof}


\begin{Prop}[\text{\cite[Fact 4.1.14]{Downarowicz11}}]\label{KS entropy is linear}
Let $(\Omega,\Sigma, \mu, f)$ be a DS. 
Then the KS entropy of $f$ is linear in time; i.e. 
\begin{equation*}
h^{KS}(f^n)=nh^{KS}(f), 
\quad\text{for all $n\in\N_0$.}
\end{equation*} 
\end{Prop}


The example of a stationary Markov process governed by the unbiased random walk on a cycle (which is defined below) is enough to show that entropy rate does not have the analogous properties of KS entropy given in Propositions~\ref{finite deterministic system} and \ref{KS entropy is linear}. 
Let $V=\{0,\ldots, N-1\}$, for some $N\in \N$ with $N\ge 3$, let $\mu$ be the uniform distribution on $V$ and consider the discrete probability space $(V,\mc P(V),\mu)$. 
The unbiased random walk on the $N$-cycle, $V$, is governed by the transition matrix $P$ with entries $p_{v+1,v}=p_{v-1,v}=1/2$, where addition is done modulo $N$, for all $v\in V$, and $p_{u,v}=0$ if $u\neq v\pm 1$.  

\begin{Prop}\label{nonzero entropy finite system}
Let $(V,\mc P(V),\mu)$ be the discrete probability space with $V=\{0,\ldots, N-1\}$, for some $N\in \N$ with $N$ odd and $N\ge 3$, $\mu$ be the uniform distribution on $V$ and $P$ be the transition matrix governing the unbiased random walk on $V$. 
Then $H(P)=\ln 2$ and $H(P^2)=\frac{3}{2}\ln 2$. 
\end{Prop}

\begin{proof}
Clearly $\mu$ is the unique probability measure that is  $P$-invariant. 
Therefore Equation~\eqref{entropy of P} gives that 
\begin{equation*}
H(P)=\sum_{v\in V}\mu_v \sum_{u\in V}\eta(p_{u, v})
=\sum_{v\in V}\frac{1}{N}2\eta(\frac{1}{2})= \ln 2.
\end{equation*} 
\noindent Also notice that, for all $v\in V$, $P^2$ has entries $p^{(2)}_{v\pm 2,v}=\frac{1}{4}$, $p^{(2)}_{v,v}=\frac{1}{2}$ and $p^{(2)}_{u,v}=0$ in all other cases, where addition is done modulo $N$. 
Again $\mu$ is the unique probability measure that is $P^2$-invariant and Equation~\eqref{entropy of P} gives that 
\begin{equation*}
H(P^2)=\sum_{v\in V}\mu_v \sum_{u\in V}\eta(p^{(2)}_{u, v}) 
= \sum_{v\in V} \frac{1}{N}(2\eta(\frac{1}{4})+\eta(\frac{1}{2}))=\frac{3}{2}\ln 2.
\end{equation*} 
\end{proof}

Proposition~\ref{nonzero entropy finite system} and Corollary~\ref{KS theorem 2} establish that the KS entropy of the symbolic dynamics of a stochastic process with range in a finite measurable space need not be zero, whereas Proposition~\ref{finite deterministic system} states that the KS entropy of a finite DS must be 0. 
Propositions~\ref{nonzero entropy finite system} and \ref{finite deterministic system} do not contradict Proposition~\ref{relationship 1} since the cardinality of $E^*$, in the symbolic dynamics of a stochastic process, is not finite unless the range, $E$, of the stochastic process is a singleton. 
Also, Proposition~\ref{nonzero entropy finite system} says that entropy rate is not linear in time whereas Proposition~\ref{KS entropy is linear} says that KS entropy is linear in time. 
Again these two propositions are not contradictory. 
We will elaborate a bit further for clarity. 
In what follows, we will denote the KS entropy of a DS $(\Omega,\Sigma,\mu, f)$ by $h^{KS}(f,\mu)$ instead of $h^{KS}(f)$ to distinguish between different measures. 
We will denote the partition dependent KS entropy similarly. 

Let $(V,\mc P(V),\mu)$ be the finite discrete probability space with $V=\{0,\ldots, N-1\}$, for some $N\in\N$ with $N$ odd and $N\ge 3$, $\mu$ be the uniform distribution on $V$, $P$ be the transition matrix governing the unbiased random walk on the $N$-cycle, $V$, $\X=(X_n)_{n=0}^\infty$ be any stationary Markov process governed by the transition matrix $P$, $\mc A_V$ the atomic partition of $V$ and $(V^*,\mc P(V)^*,\mu^{\X},s_1)$ be the symbolic dynamics of $\X$, where $s_1$ denotes the shift map on $V^*$.  
Since $\widehat{\mc A_V}$ is a generating partition for $(V^*,\mc P(V)^*,\mu^{\X},s_1)$, Corollary~\ref{KS theorem 2} shows that $H(\X)=h^{KS}(s_1,\mu^{\X})=h^{KS}(s_1,\mu^{\X}, \widehat{\mc A_V})$. 
Since $\widehat{\mc A_V}\vee s_1\inv(\widehat{\mc A_V})$ is a generating partition for $(V^*,\mc P(V)^*,\mu^{\X},s_1^2)$ with finite entropy, the KS Theorem gives that $h^{KS}(s_1^2,\mu^{\X})=h^{KS}(s_1^2,\mu^{\X}, \widehat{\mc A_V}\vee s_1\inv(\widehat{\mc A_V}))$. 
Next, consider the stationary Markov process $\Y=((X_{2n},X_{2n+1}))_{n=0}^\infty$ and let $((V\times V)^*,\mc P(V\times V)^*,\mu^{\Y},s_2)$ be the symbolic dynamics of $\Y$ and $\mc A_{V\times V}$ be the atomic partition of $V\times V$, where $s_2$ denotes the shift map on $(V\times V)^*$. 
Since $\widehat{\mc A_{V\times V}}$ is a generating partition for $((V\times V)^*,\mc P(V\times V)^*,\mu^{\Y},s_2)$ with finite entropy, Corollary~\ref{KS theorem 2} gives that $H(\Y)=h^{KS}(s_2,\mu^{\Y})=h^{KS}(s_2,\mu^{\Y},\widehat{\mc A_{V\times V}})$. 
Notice that 
\begin{equation*}
\mu^{\Y}(C\left(\begin{smallmatrix} (e_0,e_1) & \cdots & (e_{2n},e_{2n+1}) \\ 0 & \cdots & n\end{smallmatrix}\right))=
\mu^{\X}(C\left(\begin{smallmatrix} e_0 & \cdots & e_{2n+1} \\ 0 & \cdots & 2n+1\end{smallmatrix}\right)),
\end{equation*}

\noindent for all $e_0,\ldots, e_{2n+1}\in E$. 
Thus 
$$H_{\mu^{\Y}}(\vee_{k=0}^n s_2^{-k}(\widehat{\mc A_{V\times V}})))=H_{\mu^{\X}}(\vee_{k=0}^n (s_1^2)^{-k}(\widehat{\mc A_V}\vee s\inv(\widehat{\mc A_V})) \text{ for all $n\in\N_0$}$$
\noindent and therefore 
$$H(\Y)=h^{KS}(s_2,\mu^{\Y},\widehat{\mc A_{V\times V}})=h^{KS}(s_1^2,\mu^{\X}, \widehat{\mc A_V}\vee s_1\inv(\widehat{\mc A_V}))=2H(\X).$$ 
\noindent In other words, the KS entropy of $(V^*,\mc P(V)^*,\mu^{\X},s_1^2)$ is equal to the KS entropy of $((V\times V)^*,\mc P(V\times V)^*,\mu^{\Y},s_2)$ and corresponds to the entropy rate of $\Y$.  

Next consider the stochastic process $\mathbf Z=(X_{2n})_{n=0}^\infty$ and let $(V^*,\mc P(V)^*,\mu^{\mathbf Z},s_1)$ be the symbolic dynamics of $\mathbf Z$. 
Then $\mathbf Z$ is the stationary and invariant Markov process governed by the transition matrix $P^2$ and, from Proposition~\ref{nonzero entropy finite system} and Corollary~\ref{KS theorem 2}, $H(\mathbf Z)=h^{KS}(s_1,\mu^{\mathbf Z})=h^{KS}(s_1,\mu^{\mathbf Z},\widehat{\mc A_V})=\frac{3}{2}\ln 2$. 
Thus Propositions~\ref{KS entropy is linear} and \ref{nonzero entropy finite system} are not contradictory as $2H(P)=2H(\X)=h^{KS}(s_1^2,\mu^{\X})$ corresponds to the entropy rate of $\Y$, whereas $H(P^2)=H(\mathbf Z)=h^{KS}(s_1,\mu^{\Z})$ corresponds to the entropy rate of $\mathbf Z$.


\section{Measurements}\label{measurements}

Here we recall the formalism of measurements, developed by Edwards \cite{Edwards70} and Davies and Lewis \cite{DL70, Davies76}. 
We follow mainly the notations of Davies and Lewis. 
We define phase space, state space, observables and instruments. 
This formalism is general enough that it holds valid for classical mechanics, Hilbert space quantum mechanics, and $C^*$-algebra quantum mechanics (even though we don't discuss the last one here). 

A \textbf{state space} is defined as a pair $(X,K)$, where
\begin{itemize}
	\item[(i)] $X$ is a real Banach space with norm $\|\cdot \|$,
	\item[(ii)] $K$ is a closed cone in $X$,
	\item[(iii)] if $u,v\in X$, then $\|u\|+\|v\|=\|u+v\|$, and
	\item[(iv)] if $u\in X$ and $\ep>0$, then there exists $u_1, u_2\in K$ such that $u=u_1-u_2$ and $\|u_1\|+\|u_2\|<\|u\|+\ep$.
\end{itemize} 

It can be shown that, for any state space $(X,K)$, there exists a unique positive linear functional $\tau: X\ra \R$ such that $\tau(u)\le\|u\|$, for $u\in X$, with equality when $u\in K$. We say that $u\in K$ is a \textbf{state} if $\tau(u)=1$. It should be remarked that all examples of state spaces presented here will satisfy a strengthening of (iv); namely,

\begin{itemize}
	\item[(iv$'$)] if $u\in X$, then there exists $u_1, u_2\in K$ such that $u=u_1-u_2$ and $\|u_1\|+\|u_2\|=\|u\|$.
\end{itemize} 

A \textbf{phase space} is defined as an arbitrary measurable space $(\Omega, \Sigma)$, where $\Omega$ represents the collection of possible outcomes of a measurement and is sometimes called the value space in the literature. 
We say that $x:\Sigma\ra X^*$ is an \textbf{observable} if, for every $E\in\Sigma$, $0\le x(E)\le \tau$ and $x(\Omega)=\tau$, where the partial ordering on $X^*$ is defined by $\phi\le \psi$ whenever $\phi(u)\le \psi(u)$ for all $u\in K$. 
Given a state $u\in K$, an observable $x$, and $E\in\Sigma$, we interpret $x(E)u$ as the probability that a system in state $u$ takes values in $E$ when observed with the observable $x$.

An \textbf{operation} is a positive, bounded linear operator $T: X\ra X$, such that $0\le \tau(Tu)\le \tau(u)$ for every $u\in K$. 
We denote by $\mc O:= \mc O(X)$ the set of all operations on $X$. 
Finally, we define an \textbf{instrument} as a map $\mc T:\Sigma \ra \mc O$ such that $\tau(\mc T(\Omega)u)=\tau(u)$, for all $u\in K$, and  
$\mc T(\cup_n E_n)=\sum_n \mc T(E_n)$, for any disjoint sequence of sets $\{E_n\}\subseteq \Sigma$, where convergence of the sum is in the strong operator topology.

Notice that for any instrument $\mc T$, one can define a unique observable $x_{\mc T}$ by setting $x_{\mc T}(E)u=\tau(\mc T(E)u)$ for $u\in X$ and $E\in\Sigma$. 
However, it is possible that two distinct instruments, $\mc T\ne \mc S$, give rise to the same observable, $x_{\mc T}=x_{\mc S}$. 
From the above correspondence, given an initial state $u\in K$ and $E\in\Sigma$, we can interpret $\mc T(E)u/x_{\mc T}(E)u\in K$ as the state of the system immediately after measuring the system in state $u$ with the instrument $\mc T$ and obtaining values in the set $E$. 

The next example illustrates measurements in the standard version of classical mechanics. 

\begin{example}[Classical Mechanics]\label{classical mechanics}
	Let $\Omega$ be a locally compact Hausdorff space and $\mc B$ be the Borel $\sigma$-algebra of $\Omega$. 
	Take $X$ to be the real Banach space of all countably additive, regular, real-valued Borel measures on $X$ and take the norm on $X$ to be the total variation norm and the closed cone, $K$, to be the set of non-negative measures in $X$. 
	It is clear that $(X,K)$ satisfies conditions (i)-(iii) and (iv$'$) of a state space by taking $u_1$ and $u_2$ to be the positive and negative parts, respectively, of $u \in X$ given by the Hahn decomposition. Furthermore, the linear functional $\tau$ is given by $\tau(\nu)=\int_\Omega d\nu=\nu(\Omega)$ for any $\nu\in X$.
	The phase space is given by $(\Omega,\mc B)$.
	We define the (classical) \textbf{sharp measurement instrument} $\mc T$ by 
	\begin{equation}\label{sharp measurement}
	\mc T(E)\nu(A)=\nu(A\cap E) \quad \text{for $\nu\in X$ and $A,E\in \mc B$}.
	\end{equation}
	
	\noindent The corresponding observable is given by
	$$
	x_{\mc T}(E)\nu=\tau(\mc T(E)\nu)=\nu(E)\quad \text{for }E\in\mc B \text{ and } \nu\in X.
	$$	
\end{example}

The next example illustrates measurements in the Hilbert space formalism of quantum mechanics with discrete phase space. 

\begin{example}[Hilbert Space Quantum Mechanics]\label{HSM1}
	Let $H$ be a Hilbert space. 
	Take $X=S_1^{sa}(H)$ to be the real Banach space of self-adjoint, trace class operators on $H$ equipped with the trace class norm and the closed cone, $K=S_1^+(H)$, to be collection of all the positive, trace class operators on $H$. 
	It is clear that the state space $(X,K)$ satisfies conditions (i)-(iii) and (iv$'$) of a state space by taking $u_1$ and $u_2$ to be the positive and negative parts, respectively, of $u \in X$ given by the functional calculus. 
	Furthermore, the linear functional $\tau$ is given by  the trace, $\tr$. 
	Let $(\Omega,\mc P(\Omega))$ be a discrete phase space and  $(B_i)_{i\in\Omega}\subseteq B(H)$ be a collection of bounded operators, indexed by $\Omega$, such that $\sum_{i\in \Omega} B_i^* B_i=\1$, where $\1$ is the identity operator on $H$.   
	We define the instrument $\mc T$, given by the family $(B_i)_{i\in\Omega}$, by 
	$$
	\mc T(E)\rho=\sum_{i\in E}B_i\rho B_i^*\quad\text{for each $\rho\in X$ and $E\in \mc P(\Omega)$,}
	$$
	
	\noindent where the sums are taken with respect to the strong operator topology if $\Omega$ is countably infinite. 
	When restricted to actions on $K$, $\mc T$ represents a positive operator valued measure. 
	The corresponding observable is given by
	$$
	x_{\mc T}(E)\rho =\sum_{i\in E}\tr(B_i \rho B_i^*)\quad\text{for each $\rho\in X$ and $E\in \mc P(\Omega)$.}
	$$
\end{example} 	
	
Whenever the family of bounded operators are pairwise orthogonal projections, we denote them by $(P_i)_{i\in\Omega}$ and note that  $\sum_{i\in\Omega}P_i^*P_i=\sum_{i\in\Omega}P_i=\1$. 
In this case, the corresponding instrument $\mc T$, for $(P_i)_{i\in\Omega}$, is called a \textbf{L\"uders-von Neumann instrument} and is given by 
\begin{equation*}
\mc T(E)\rho=\sum_{i\in E}P_i \rho P_i \quad \text{for $\rho \in X$ and $E\in\mc P(\Omega)$,}
\end{equation*}

\noindent where the sums are taken with respect to the strong operator topology if $\Omega$ is countably infinite. 	
\noindent It is worth noting that $\mc T$ is defined by the ``collapse of wave function formula." 
The corresponding observable is defined analogously. 
 
Whenever the family $(P_i)_{i\in\Omega}$ consists of orthogonal, rank-1 projections, the L\"uders-von Neumann instrument $\mc T$ is called a \textbf{coherent states instrument} (see \cite[Section IV]{SZ94}.) 
In this paper whenever we refer to a coherent states instrument we will always mean a L\"uders-von Neumann instrument given by a family of orthogonal, rank-1 projections as opposed to the more general definition given in \cite[Example~(M)]{SZ94}.


\section{Quantum Dynamical Entropy}\label{Dyn Ent of UQRWs}


There have been many successful attempts to generalize KS entropy to a quantum dynamical entropy in \cite{CNT87, AF94, SZ94, AOW97, KOW99} and more. 
In \cite{SZ94}, the authors use a semi-classical approach to develop a quantum dynamical entropy using the general notions of state space, phase space, observables and instruments introduced in Section~\ref{measurements}. 
Furthermore, the quantum dynamical entropy of \cite{SZ94} has the benefit that it is not guaranteed to be zero for finite systems, unlike the others. 
This is due to its symbolic dynamics approach.  

Let $(X,K)$ be a state space and $u\in K$ be a state. 
Let $(\Omega, \Sigma)$ be a phase space, $\mc T$ an instrument and $\Theta$ a $\tau$-preserving automorphism of $X$; i.e. $\tau(\Theta v)=\tau(v)$ for all $v\in X$. 
Let $(\Omega^*,\Sigma^*)$ be the measurable space defined in Section~\ref{symbolic dynamics}. 
We will define an instrument and state-dependent probability measure, $\mu^{(\Theta, \mc T, u)}$, on $(\Omega^*,\Sigma^*)$. 
First, we define the values of $\mu^{(\Theta, \mc T, u)}$ on the cylinder sets in $\Sigma^*$ with an initial interval of time sequences, $\{k\}_{k=0}^n$ for some $n\in\N_0$,  by  
\begin{equation}\label{general entropy definition}
\mu^{(\Theta,\mc T, u)}(C\left(\begin{smallmatrix} A_0 & \cdots & A_{n} \\ 0 & \cdots & n\end{smallmatrix}\right))
= \tau(\mc T(A_n)\circ \Theta\circ \cdots \circ \mc T(A_1)\circ \Theta\circ \mc T(A_0)u),
\end{equation}

\noindent for all $A_0,\ldots, A_n\in \Sigma$. 
Since the collection of cylinder sets with an initial interval of time sequences form a $\pi$-system which generates $\Sigma^*$, there is a unique extension of  $\mu^{(\Theta, \mc T, u)}$ to $(\Omega^*,\Sigma^*)$ by the $\pi$-$\lambda$ Theorem. 

Notice that, for a stochastic process $\X$, we defined the measure $\mu^{\X}$ first on cylinder sets with arbitrary time sequences (Equation~\eqref{symbolic dynamics measure}), whereas the measure, $\mu^{(\Theta,\mc T, u)}$ in Equation~\eqref{general entropy definition}, was defined first on the cylinder sets with an initial interval of time sequences. 
By defining $\mu^{(\Theta,\mc T, u)}$ in this way, we have that, for $A_0,A_2\in \Sigma$,
\begin{equation*}
\mu^{(\Theta,\mc T, u)}(C\left(\begin{smallmatrix} A_0 &  A_{2} \\ 0 & 2\end{smallmatrix}\right))
= \mu^{(\Theta,\mc T, u)}(C\left(\begin{smallmatrix} A_0 &\Omega &  A_{2} \\ 0 &1 & 2\end{smallmatrix}\right)) 
= \tau(\mc T(A_2)\circ\Theta\circ \mc T(\Omega)\circ \Theta\circ \mc T(A_0)u),
\end{equation*}

\noindent which is not necessarily equal to $\tau(\mc T(A_2)\circ\Theta^2 \circ \mc T(A_0)u)$. 
Therefore $\mu^{(\Theta,\mc T, u)}(C\left(\begin{smallmatrix} A_0 &  A_{2} \\ 0 & 2\end{smallmatrix}\right))$ is interpreted as the probability that a system in initial state $u$ will be measured at times $0,1, 2$ and record the measurement sequence $(A_0,A_2)$ at times $0$ and $2$. 
In other words, we must assume that the instrument $\mc T$ is interacting with the system at all integer times, regardless of whether or not we record a measurement. 

%

Define the $(\Omega^*,\Omega)$ stochastic process $\X^{(\Theta,\mc T, u)}=(X^{(\Theta,\mc T, u)}_n)_{n=0}^\infty$ by $X^{(\Theta,\mc T, u)}_n(x)=x_n$ for each $x=(x_m)_{m\in\N_0}\in \Omega^*$ and, for each $\mc C\in\pa(\Omega)$, define the $(\Omega^*,\mc C)$ stochastic process $\X^{(\Theta,\mc T, u)}_{\mc C}=(X^{(\Theta,\mc T, u,\mc C)}_n)_{n=0}^\infty$ by $X^{(\Theta,\mc T, u,\mc C)}_n=i_{\mc C}\circ X^{(\Theta,\mc T, u)}_n$, where $i_{\mc C}:\Omega\ra \mc C$ is the natural map that assigns to each $x\in \Omega$ the unique $A\in \mc C$ such that $x\in A$. 
Even though the formulas of $\X^{(\Theta,\mc T,u)}$ and $\X^{(\Theta,\mc T,u)}_{\mc C}$ do not depend on $\Theta$, $\mc T$ and $u$, the measure $\mu^{(\Theta,\mc T,u)}$ on their domain, $\Omega^*$, depends on $\Theta$, $\mc T$ and $u$.

We define the \textbf{S\l{}omczy\'nski-\.Zyczkowski (SZ) entropy of \bm{$(\Theta,\mc T, u)$}
with respect to \bm{$\mc C$}} to be the entropy rate of the stochastic process $\X^{(\Theta,\mc T, u)}_{\mc C}$. It is related to $\mu^{(\Theta,\mc T, u)}$ by the equation 
\begin{equation}\label{SZ entropy}
h^{SZ}(\Theta,\mc T, u,\mc C):= H(\X^{(\Theta,\mc T, u)}_{\mc C})=\lim_{n\ra\infty}\frac{1}{n} \sum_{\substack{ A_{k}\in \mc C \\ 0\le k\le n-1}}
\eta(\mu^{(\Theta,\mc T, u)}(C\left(\begin{smallmatrix} A_{0} & \cdots & A_{n-1} \\ 0 & \cdots & n-1\end{smallmatrix}\right))),
\end{equation}

\noindent whenever the limit exists. 
The second equality follows from Equations~\eqref{static entropy of process} and \eqref{entropy rate}. 

\begin{remark}\label{entropy equivalence 3}
Let $s$ be the shift transformation on $(\Omega^*,\Sigma^*,\mu^{(\Theta, \mc T, u)})$ so that $(\Omega^*,\Sigma^*,\mu^{(\Theta, \mc T, u)}, s)$ is a DS. 
From Proposition~\ref{relationship 1} and the definition of $\X^{(\Theta,\mc T, u)}_{\mc C}$ it is clear that $$h^{SZ}(\Theta,\mc T, u,\mc C)=H(\X^{(\Theta,\mc T, u)}_{\mc C})=h^{KS}(s,\widehat{\mc C}).$$ 
\end{remark}

Next we split the SZ entropy of $(\Theta,\mc T, u)$ with respect to $\mc C$ into two different causes for randomness. 
The first cause of randomness is that caused by the choice of instrument, is referred to as the \textbf{measurement SZ entropy} and is given by 
\begin{equation}\label{measurement entropy}
h^{SZ}_{\text{meas}}(\mc T, u, \mc C): = h^{SZ}(\1, \mc T, u, \mc C).
\end{equation}

\noindent The second cause of randomness is given by the dynamics; i.e. the automorphism $\Theta$, is referred to as the \textbf{dynamical SZ entropy}, and is given by the difference  
\begin{equation*}
h^{SZ}_{\text{dyn}}(\Theta, \mc T, u, \mc C): = h^{SZ}(\Theta, \mc T, u, \mc C)- h^{SZ}_{\text{meas}}(\mc T, u, \mc C).
\end{equation*}

Finally, we define the \textbf{dynamical SZ entropy of \bm{$(\Theta,\mc T, u)$}} by 
\begin{equation}\label{SZ entropy 2}
h^{SZ}_{\text{dyn}}(\Theta,\mc T,u):= \sup_{\substack{\mc C\in\pa(\Omega)}\\ H(\widehat{\mc C})<\infty} h^{SZ}_{\text{dyn}}(\Theta,\mc T,u, \mc C).
\end{equation}

The following lemma 
is claimed in \cite{SS17}. 
For completeness we provide the proof. 

\begin{lemma}\label{lvn gives 0 meas entropy}
Let $(\Omega,\mc P(\Omega))$, $(X,K)$ and $H$ be as in Example~\ref{HSM1} and let $\mc T$ be a L\"uders-von Neumann instrument. 
Then $h^{SZ}_{\text{meas}}(\mc T, \rho, \mc C)=0$ for any state $\rho\in K$ and any $\mc C\in \pa(\Omega)$ with finite entropy; i.e. $H(\widehat{\mc C})<\infty$. 
\end{lemma}

\begin{proof}
Let $(P_i)_{i\in\Omega}$ be the family of pairwise orthogonal projections that governs $\mc T$ and fix a state $\rho\in K$.
Since the family, $(P_i)_{i\in\Omega}$, is pairwise orthogonal we have, for any $n\in\N_0$ and $A_{0},\ldots, A_{n}\in \mc P(\Omega)$, that 
\begin{equation*}
\mu^{(\1,\mc T, \rho)}(C\left(\begin{smallmatrix} A_{0} & \cdots & A_{n} \\ 0 & \cdots & n\end{smallmatrix}\right))
= \begin{cases} \sum_{a\in A_0} \tr(P_a\rho P_a)=\mu^{(\1,\mc T,\rho)}(\left(\begin{smallmatrix} A_{0} \\ 0\end{smallmatrix}\right)) & \text{if }A_{0}=\cdots=A_{n} \\ 0 & \text{else}\end{cases}
\end{equation*}

\noindent Therefore, for any $\mc C\in \pa(\Omega)$ with $H(\widehat{\mc C})<\infty$, we have  
\begin{equation*}
\begin{split}
h^{SZ}_{\text{meas}}(\mc T, \rho, \mc C)&=\lim_{n\ra\infty} \frac{1}{n} \sum_{A\in \mc C}\eta (\mu^{(\1,\mc T, \rho)}(C\left(\begin{smallmatrix} A  \\ 0\end{smallmatrix}\right))) 
\quad\text{by \eqref{SZ entropy} and \eqref{measurement entropy}} \\
&= \lim_{n\ra\infty}\frac{1}{n} H(\widehat{\mc C})=0
\quad\text{by the definition of $\widehat{\mc C}$.}
\end{split}
\end{equation*}

%
\end{proof}

\begin{remark}
It is natural to consider only the partitions $\mc C\in\pa(\Omega)$ with finite entropy in Lemma~\ref{lvn gives 0 meas entropy}, because these are the only partitions considered in Equation~\eqref{SZ entropy 2}.  
\end{remark}

Fix a discrete phase space $(\Omega,\Sigma)$ with $|\Omega|= N$.  
Then Lemma~\ref{lvn gives 0 meas entropy}, together with Equation~\eqref{upper bound}, implies that $h^{SZ}_{\text{dyn}}(\Theta, \mc T, \rho, \mc C)= h^{SZ}(\Theta, \mc T, \rho, \mc C)$ for any unitary transformation $\Theta$, partition $\mc C$, state $\rho$, and coherent states instrument $\mc T$. 
Next we show that the measurement SZ entropy for classical sharp measurement instruments is equal to 0 as well. 

\begin{lemma}\label{csm gives 0 meas entropy}
Let $(\Omega,\mc B)$, $(X,K)$, $\tau$ and $\mc T$ be as in Example~\ref{classical mechanics}. 
Then, for any state $\mu\in K$ and partition $\mc C\in\pa(\Omega)$ with finite entropy, we have $h^{SZ}_{\text{meas}}(\mc T, \mu, \mc C)=0$.  
\end{lemma}

\begin{proof}
Fix a state $\mu\in K$ and a partition $\mc C\in\pa(\Omega)$ with finite entropy. 
Then, for any $n\in\N_0$ and $A_{0},\ldots, A_{n}\in \mc C$, we have that 
\begin{equation*}
\mu^{(\1,\mc T, \mu)}(C\left(\begin{smallmatrix} A_{0} & \cdots & A_{n} \\ 0 & \cdots & n\end{smallmatrix}\right))
= \begin{cases} \mu(A_0) & \text{if }A_{0}=\cdots=A_{n} \\ 0 & \text{else}\end{cases}
\end{equation*}

\noindent Therefore  
\begin{equation*}
\begin{split}
h^{SZ}_{\text{meas}}(\mc T, \mu, \mc C)&=\lim_{n\ra\infty} \frac{1}{n} \sum_{A\in \mc C}\eta (\mu(A)) \\
&= \lim_{n\ra\infty}\frac{1}{n} H(\mc C)=0.
\end{split}
\end{equation*}
\end{proof}

The following result is claimed without proof in \cite[Proposition~4(A)]{SZ94} and states that the KS and SZ entropies agree for classical dynamical systems with sharp instruments. 

\begin{Prop}
	Let $(\Omega,\mc B)$, $(X,K)$, $\tau$ and $\mc T$ be as in Example~\ref{classical mechanics}. 
	Let $\mu\in K$ be a state; i.e. a probability measure on $(\Omega,\mc B)$, and $f:\Omega\ra \Omega$ a measurable map so that $(\Omega,\mc B,\mu,f)$ is a DS. 
	Let $T_f: X\ra X$ be the automorphism known as the \textbf{Koopman operator} defined by 
	\begin{equation*}
	T_f(\nu)(A):= \nu(f\inv(A))
	\quad\text{for all }\nu\in X\text{ and }A\in \mc B. 
	\end{equation*}
	
	\noindent Then for each $\mc C\in \pa(\Omega)$, $h^{KS}(f,\mc C)=h^{SZ}_{\text{dyn}}(T_f,\mc T, \mu, \mc C)$. 
\end{Prop}

\begin{proof} 
	Fix a partition $\mc C\in\pa(\Omega)$. 
	For all $n\in\N_0$ and $A_{0},\ldots,A_{n}\in \mc C$ we see that 
	\begin{align*}
	\mu^{(T_f,\mc T, \mu)}(C\left(\begin{smallmatrix} A_{0} & \cdots & A_{n} \\ 0 & \cdots & n\end{smallmatrix}\right))
	&= \tau(\mc T(A_{n})\circ T_f\circ \cdots\circ T_f\circ \mc T(A_{0}) \mu) 
	\quad\text{by \eqref{general entropy definition}} \\
	&=( \mc T(A_n)\circ T_f\circ \cdots\circ T_f\circ \mc T(A_{0})) \mu(X)
	\quad\text{by definition of $\tau$} \\
	&= ( \mc T(A_n)\circ T_f\circ \cdots\circ \mc T(A_{1})\circ T_f) \mu(A_0)
	\quad\text{by \eqref{sharp measurement}} \\
	&= ( \mc T(A_n)\circ T_f\circ \cdots\circ T_f\circ \mc T(A_{1})) \mu(f\inv (A_0)) 
	\quad\text{by \eqref{koopman operator}} \\
	&= \cdots \\
	&= \mu(A_{i_n}\cap f\inv(A_{i_1}) \cap \cdots \cap f^{-n}(A_{i_0})).
	\end{align*}
	
	\noindent Using Remark~\ref{pullbacks} and Lemma~\ref{csm gives 0 meas entropy}, we get $$h^{KS}(f,\mc C)=h^{SZ}(T_f,\mc T, \mu, \mc C)=h^{SZ}_{\text{dyn}}(T_f,\mc T, \mu, \mc C).$$ 
\end{proof}


Next we examine the properties of SZ entropy for coherent states instruments. 
Let $H$ be a Hilbert space and $(X,K)$ the state space defined in Example~\ref{HSM1}. 
Given a unitary operator, $U$, on $H$, the \textbf{unitary transformation, \bm{$\Theta:X\ra X$}, of \bm{$U$}} is given by 
\begin{equation*}
\Theta(\cdot)=U\cdot U^*.
\end{equation*}

\noindent The following lemma gives a simplification of Equation~\eqref{general entropy definition} for coherent states instruments. 
The moreover statement of the following lemma is mentioned in \cite[page 3]{SS17} without proof.

\begin{lemma}\label{cs probabilities} 
Let $H$, $(X,K)$, $\tau$ and $(\Omega,\mc P(\Omega))$ be as in Example~\ref{HSM1}. 
Let $(P_i)_{i\in\Omega}$ be a family of orthogonal, rank-1 projections on $H$ such that $\sum_{i\in\Omega} P_i=\1$, and let $a_i\in H$ such that $P_i=|a_i\rangle\langle a_i|$, for each $i\in\Omega$. 
Let $\mc T$ be the coherent states instrument for $(P_i)_{i\in\Omega}$, $U$ a unitary operator on $H$, $\Theta$ the unitary transformation of $U$ and $\rho\in K$ a state.
Then, for all $n\in\N_0$ and $A_{0},\ldots,A_{n}\in\mc P(\Omega)$, 
\begin{equation}\label{cs probabilities 2}
\mu^{(\Theta,\mc T, \rho)}(C\left(\begin{smallmatrix} A_{0} & \cdots & A_{n} \\ 0 & \cdots & n\end{smallmatrix}\right))
=\sum_{\substack{a_k\in A_{k}\\ 0\le k\le n}} 
\langle a_0|\rho|a_0\rangle\prod_{k=1}^n |\langle a_k|U|a_{k-1}\rangle|^2.
\end{equation}

\noindent Moreover, $\X^{(\Theta,\mc T,\rho)}$ is a Markov process governed by the transition matrix $P$ on $\Omega$ with $(i,j)$-entry given by $|\langle a_i|U|a_j\rangle|^2$, for all $i,j\in\Omega$. 
\end{lemma} 

\begin{proof}
By direct calculation, Equation~\eqref{general entropy definition} simplifies to 
\begin{align*}
\mu^{(\Theta,\mc T, \rho)}(C\left(\begin{smallmatrix} A_{0} & \cdots & A_{n} \\ 0 & \cdots & n\end{smallmatrix}\right))
&= \tau(\mc T(A_{n})\circ \Theta\circ\cdots\circ  \Theta\circ \mc T(A_{0}) \rho) \\
&= \sum_{\substack{a_k\in A_{k}\\ 0\le k\le n}}
\tr(\mc T(\{a_{n}\})\circ \Theta\circ\cdots\circ  \Theta\circ \mc T(\{a_{0}\}) \rho) \\
&= \sum_{\substack{a_k\in A_{k}\\ 0\le k\le n}}
\tr(P_{a_n}U\cdots U P_{a_0} \rho P_{a_0}U^*\cdots U^* P_{a_n}) \\
&= \sum_{\substack{a_k\in A_{k}\\ 0\le k\le n}}
\tr(|a_n\rangle\langle a_n|U\cdots U |a_0\rangle\langle a_0| \rho |a_0\rangle\langle a_0|U^*\cdots U^* |a_n\rangle\langle a_n|) \\
&= \sum_{\substack{a_k\in A_{k}\\ 0\le k\le n}} 
\langle a_0|\rho|a_0\rangle\prod_{k=1}^n |\langle a_k|U|a_{k-1}\rangle|^2,
\end{align*}

\noindent where the second to last equality follows from writing $P_{a_k}=|a_k\rangle\langle a_k|$, for all $0\le k\le n$ and the last equality follows since $\overline{\langle a_{k-1}|U^*|a_k\rangle}=\langle a_k|U|a_{k-1}\rangle$, for all $1\le k\le n$. 
It is immediately clear that $\X^{(\Theta,\mc T,\rho)}$ is a stationary Markov process governed by the transition matrix $P$. 
\end{proof} 
 
It is worth noting that Equation~\eqref{cs probabilities 2} is a simplification of the probabilities in \cite[Equations (27)-(29)]{SZ94} for L\"uders-von Neumann coherent states instruments. 

\begin{corollary}\label{markov entropy rate 3}
Let $H$, $(X,K)$, $\tau$ and $(\Omega,\mc P(\Omega))$ be as in Example~\ref{HSM1}. 
Let $(P_i)_{i\in\Omega}$ be a family of orthogonal, rank-1 projections on $H$ such that $\sum_{i\in\Omega} P_i=\1$ and $P_i=|a_i\rangle\langle a_i|$ for some $a_i\in H$ and each $i\in\Omega$, $\mc T$ the coherent states instrument for $(P_i)_{i\in\Omega}$, $U$ a unitary operator on $H$, $\Theta$ the unitary transformation of $U$, $\rho\in K$ a state and $P$ the transition matrix defined in Lemma~\ref{cs probabilities}. 
Then 
\begin{equation*}
h^{SZ}(\Theta, \mc T, \rho, \mc A)=\lim_{n\ra\infty}
\sum_{y\in \Omega} (P^n\mu)_y\sum_{x\in \Omega}\eta(|\langle a_x|U|a_y\rangle|^2),
\end{equation*}

\noindent where $\mu=p_{X^{(\Theta,\mc T,\rho)}_0}$ and $\mc A$ is the atomic partition of $\Omega$. 
Moreover, whenever $\mu=(\mu_y)_{y\in\Omega}$ is $P$-invariant, we have $h^{SZ}(\Theta, \mc T, \rho, \mc A)= \sum_{y\in \Omega} \mu_y\sum_{x\in \Omega}\eta(|\langle a_x|U|a_y\rangle|^2)$. 
\end{corollary}

\begin{proof}
This follows immediately from Lemma~\ref{cs probabilities} and Theorem~\ref{markov entropy rate}. 
\end{proof}

In \cite[Section IV]{SZ94} the authors require that the state $\rho$ is invariant in the sense that 
\begin{equation}\label{SZ invariance}
\Theta(\mc T(\Omega)\rho)=\rho
\end{equation}

\noindent when defining SZ entropy for coherent states instruments. 
This seems to be due to the fact that, in \cite[Proposition 2(B)]{SZ94}, the authors show that, under Assumption~\eqref{SZ invariance}, for a general coherent states instrument, the stochastic process $X^{(\Theta,\mc T,\rho)}$ is stationary and hence, by the ``moreover" part of Corollary~\ref{markov entropy rate 3}, $h^{SZ}(\Theta, \mc T, \rho, \mc A)= \sum_{y\in \Omega} \mu_y\sum_{x\in \Omega}\eta(|\langle a_x|U|a_y\rangle|^2)$. 
We find Assumption~\eqref{SZ invariance} restrictive and do not adopt it here. 
It is also worth mentioning that another invariance condition often imposed on quantum dynamical systems is that $\Theta(\rho)=\rho$. 
For instance, for AFL entropy in \cite[Page 76]{AF94} which is formulated for general $C^*$-algebras, the authors require that a state $\omega$ satisfies $\omega\circ\Theta=\omega$ which is equivalent to $\Theta(\rho)=\rho$ in the Hilbert space quantum mechanics picture whenever $\omega$ is defined by $\omega(\cdot)=\tr(\rho \cdot)$. 
Also, given a Hilbert space $H$, a unitary operator $U$ on $H$ and a norm-1 eigenvector, $x\in H$, of $U$,  the pure (or vector) state $\rho=|x\rangle\langle x|$ satisfies $\Theta(\rho)=\rho$, where $\Theta$ is the unitary transformation of $U$. 
There has been a lot of interest in finding these pure, invariant states in the literature for unitary quantum random walks (see e.g. \cite{KLS13, EK14a}). 
Therefore $\Theta(\rho)=\rho$ seems another natural definition of invariance. 
However, we will show in Proposition~\ref{CS entropy for Hadamard} that $\Theta(\rho)=\rho$ does not imply that $X^{(\Theta,\mc T,\rho)}$ is an invariant stochastic process. 

The following result states that SZ entropy is not linear in the time interval between successive measurements which answers an open problem posed in \cite[page 5692 Question~(2)]{SZ94}. 
This result is in contrast to KS entropy which is linear in time (see Proposition~\ref{KS entropy is linear}). 
Moreover, since entropy rate is nonlinear in time (see Proposition~\ref{nonzero entropy finite system}), the result gives further evidence that measurements of a deterministic quantum system produce properties that are probabilistic in nature. 

\begin{theorem}\label{SZ entropy is nonlinear}
Let $(X,K)$ be as in Example~\ref{HSM1}. 
Let $(\Omega,\mc P(\Omega))$ be a discrete phase space with $|\Omega|=N$ for some $N\in\N$, $\mc T$ a L\"uders-von Neumann instrument, $\Theta$ a unitary transformation and $\rho\in K$ a state. 
Then $h^{SZ}_{\text{dyn}}(\Theta^n, \mc T, \rho)\le N$ for all $n\in\N$. 
Therefore, if $h^{SZ}_{\text{dyn}}(\Theta, \mc T, \rho)\ne 0$, then $h^{SZ}_{\text{dyn}}(\Theta^n, \mc T, \rho)\ne nh^{SZ}_{\text{dyn}}(\Theta, \mc T, \rho)$ for all sufficiently large $n\in\N$. 
\end{theorem}

\begin{proof}
Let $(\Omega^*,\mc P(\Omega)^*, \mu^{(\Theta^n,\mc T, \rho)}, s)$ be the symbolic dynamics of $(\Theta^n ,\mc T, \rho)$ for each $n\in \N$ and let $\mc A$ be the atomic partition of $\Omega$. 
Using Equation~\eqref{upper bound}, we have that $H_{\mu^{(\Theta^n,\mc T, \rho)}}(\vee_{k=0}^{m-1} s^{-k}(\widehat{\mc A}))\le \ln(|\vee_{k=0}^{m-1} s^{-k}(\widehat{\mc A})|)\le m\ln N$, for all $n,m\in\N$. 
Therefore $h^{SZ}_{\text{dyn}}(\Theta^n, \mc T, \rho, \mc A)\le \ln N$, for each $n\in\N$.  
If $h^{SZ}_{\text{dyn}}(\Theta, \mc T, \rho, \mc A)=k\ne 0$, then, since $\widehat{\mc A}$ is a generating partition for $(\Omega^*,\mc P(\Omega)^*, \mu^{(\Theta^n,\mc T, \rho)}, s)$, we have, for all $n> \frac{\ln N}{k}$, that 
$$
h^{SZ}_{\text{dyn}}(\Theta^n,\mc T,\rho)=
h^{SZ}_{\text{dyn}}(\Theta^n, \mc T, \rho, \mc A)\le \ln N< nh^{SZ}_{\text{dyn}}(\Theta, \mc T, \rho, \mc A)=
nh^{SZ}_{\text{dyn}}(\Theta, \mc T, \rho).
$$ 
\end{proof}

In \cite{SS17} the authors establish a class of instruments which have positive dynamical SZ entropy and we give further such examples in Section~\ref{UQRW dynamical entropy}. 
Therefore Proposition~\ref{SZ entropy is nonlinear} does establish the nonlinearity of dynamical SZ entropy in time. 
We illustrate this fact in Section~\ref{UQRW dynamical entropy} by calculating the SZ entropy of the Hadamard walk with a L\"uders-von Neumann instrument given by a family of rank-2 projections (Theorem~\ref{coarser instrument}). 
The following section is dedicated to defining unitary quantum random walks and, in particular, the Hadamard walk.


\section{Unitary Quantum Random Walks}\label{Hadamard walk section}

The unitary quantum random walk (UQRW) is one of the many adaptations of the classical random walk to the quantum domain and, in particular, is the adaptation of classical random walks for closed quantum systems.  
We will define the UQRW on a finite or countably infinite vertex set $V$. 
To consider a collection of vertices in the quantum domain, Hilbert space quantum mechanics is used (see Example~\ref{HSM1}) and we consider the \textbf{position space}, $H_P:=\ell_2(V)$, with an orthonormal basis, $\{|v\rangle\}_{v\in V}$, indexed by $V$. 
To add internal degrees of freedom to the vertices, the \textbf{coin space} $H_C$ is used, which is an at most countably dimensional Hilbert space. 
In general, a UQRW is given by the unitary transformation over the tensored Hilbert space $H=H_C\ten H_P$. 
The most common UQRWs are the so-call coined UQRWs. 
To define these we must first fix an orthonormal basis, $\{|c,v\rangle\}_{(c,v)\in C\times V}$ on $H$, sometimes referred to as the computational basis, for some index set $C$, where $|c, v\rangle = |c\rangle\ten |v\rangle$. 
We say that a UQRW is \textbf{coined} if it is the unitary transformation of an operator $U$ of the form 
\begin{equation}\label{UQRW}
U=S(\sum_{v\in V} U_v\ten |v\rangle\langle v|),
\end{equation}

\noindent where $U_v$ is a unitary operator on $H_C$ for each $v\in V$, and $S$ is a permutation operator which is referred to as the \textbf{shift operator}. 
By a ``permutation operator", we mean that $S$ has the form  
\begin{equation}\label{shift}
S=\sum_{(c,v)\in C\times V} |\sigma(c,v)\rangle \langle c, v|,
\end{equation}

\noindent for some permutation $\sigma$ of $C\times V$. 
For each $v\in V$, the unitary operator $U_v$, referred to as the coin operator at $v$, changes the coin state at the vertex $v$ in a deterministic way while the shift operator $S$ moves the random walker from one site to another. 
In the sequel we will only consider coined UQRWs and we will drop the adjective. 
We say that a UQRW is \textbf{space homogeneous} if there exists a unitary operator $W$ on $H_C$ such that $W=U_v$ for every $v\in V$. 
For a space homogeneous UQRW the form of $U$ in Equation~\eqref{UQRW} simplifies to 
\begin{equation*}
U=S(W\ten \1_{H_P}).
\end{equation*}
 
We say that a shift operator $S$ is \textbf{coin preserving} if, for each $c\in C$, there exists a permutation $\sigma_c$ of $V$ such that 
\begin{equation*}
S=\sum_{(c,v) \in C\times V} |c,\sigma_c(v)\rangle\langle c, v|.
\end{equation*}

\noindent Furthermore, we say that a UQRW is coin preserving if its shift operator is coin preserving. 
Notice that a coin preserving shift operator moves the random walker from site to site without affecting the internal state of the walker. 

Let $U$ be a unitary operator of the form Equation~\eqref{UQRW} with shift operator, $S$, given by Equation~\eqref{shift}. 
To draw a connection to classical random walks it is helpful to visualize a random walker on the directed graph $G=(V,E)$ where $E$ is the edge set determined by the shift operator $S$. 
That is to say, for all $u,v\in V$, $(u,v)\in E$ if and only if there exists $c_1,c_2\in C$ such that $\sigma(c_1,u)=(c_2,v)$, where $\sigma$ is the map appearing in Equation~\eqref{shift}, or, equivalently, $P_v SP_u\ne 0$, where $P_v=\1_{H_C}\ten |v\rangle\langle v|$; i.e. $P_v$ is the projection from $H$ to $H_C\ten \text{span}(v)$. 


We are specifically interested in the Hadamard walk, which has been studied extensively in the literature (e.g. \cite{ABNVW01, Kempe03, Portugal13}) and is defined below. 
Consider the vertex set $V=\{0,\ldots, N-1\}$, for some $N\in\N$ with $N\ge 2$, and set $H_P=\C^N$. 
Let $H_C=\C^2$, with orthonormal basis $\{|R\rangle, |L\rangle\}$. 
Define the (coin preserving) integer shift operator by 
\begin{equation*}
S=\sum_{n=0}^{N-1} |R,n+1\rangle\langle R,n| + |L,n-1\rangle\langle L,n|,
\end{equation*}
 
\noindent where addition on the integers is done modulo $N$. 
Throughout the rest of the paper addition (on $V$) will be done modulo $N$. 
Notice that $|R\rangle$ now corresponds to a shift right on the integers and $|L\rangle$ corresponds to a shift left on the integers.  
In this case the directed graph $G=(V,E)$ has edge set which is given by  $E=\{(n,n+1),(n,n-1)\}_{n=0}^{N-1}$. 
The unitary operator 
\begin{equation*}
h:=\frac{1}{\sqrt 2} \begin{bmatrix} 1 & 1 \\ 1 & -1\end{bmatrix},
\end{equation*}

\noindent on $H_C$ is referred to as the Hadamard matrix (or Hadamard coin/gate). 
The \textbf{Hadamard walk on \bm{$V$}} is the map $\Theta:X\ra X$, where $(X,K)$ is the state space defined in Example~\ref{HSM1}, given by 
\begin{equation}\label{Hadamard walk}
\Theta(\rho)=U\rho U^*, \text{ for each }\rho\in X, 
\text{ where }U=S(h\ten \1_{H_P}).
\end{equation}

\noindent It is clear that the Hadamard walk is coin preserving and space homogeneous. 

\begin{remark}
The Hadamard walk can easily be extended to $V=\Z$ as opposed to the finite cycle $V=\{0,\ldots, N-1\}$ and it is often viewed in this manner. (e.g. \cite[Section 5.1]{Portugal13}) 
\end{remark}


\section{SZ Entropy of the Hadamard walk}\label{UQRW dynamical entropy}

Let $H=H_C\ten H_P$, $C=\{R,L\}$ and take the phase space to be $(C\times V,\mc P(C\times V))$. 
Take $(X,K)$ be the state space $(S_1^{sa}(H),S_1^+(H))$ as in Example~\ref{HSM1}. 
Define the family of orthogonal, rank-1 projections $(P_e)_{e\in C\times V}$ by 
\begin{equation}\label{cs projections}
P_e=|c,v\rangle\langle c,v|, \text{ whenever }e=(c,v)\in C\times V.
\end{equation} 

\noindent Let $\mc T$ be the coherent states instrument governed by the family $(P_e)_{e\in C\times V}$. 
The next proposition states that, for a unitary transformation $\Theta$ and a state $\rho\in K$, $\Theta(\rho)=\rho$ does not imply that the associated Markov chain is stationary. 
The result shows that this natural definition of invariance for $\rho$ is not sufficient for stationarity, whereas Assumption~\eqref{SZ invariance}, imposed by the authors of \cite{SZ94}, does guarantee stationarity. 

\begin{Prop}\label{CS entropy for Hadamard}
Let $\Theta$ be the Hadamard walk on $V$ with $|V|=N\ge 2$. 
Let $\mc T$ be the coherent states instrument given by the family of orthogonal projections $(P_e)_{e\in C\times V}$ and $\mc A$ be the atomic partition on $C\times V$. 
Let $|x\rangle=\frac{1}{\sqrt{N(4+2\sqrt{2})}}((1+\sqrt 2)|R\rangle +|L\rangle)\ten \sum_{v\in V} |v\rangle$ (which is a unit norm eigenvector for the unitary matrix, $U$, of the Hadamard walk) and $\rho=|x\rangle\langle x|$. 
Then the pmf, $p_{X^{(\Theta,\mc T,\rho)}_0}$, of $X^{(\Theta,\mc T,\rho)}_0$ is not $P$-invariant, where $P$ is the transition matrix defined in Lemma~\ref{cs probabilities}. 
Furthermore, the dynamical SZ entropy is equal to $h^{SZ}_{\text{dyn}}(\Theta, \mc T,\rho)=\ln 2$.
\end{Prop}

\begin{proof}
For each $(c,v)\in C\times V$, 
\begin{equation}\label{equation0}
p_{X^{(\Theta,\mc T,\rho)}_0}(c,v)=\langle c,v|\rho|c,v\rangle = \frac{1}{N(4+2\sqrt{2})}((3+2\sqrt 2)\delta_{c,R}+\delta_{c,L}).
\end{equation}

\noindent Also, for each $e=(c,v), f=(d,u)\in C\times V$, a straightforward calculation yields  
\begin{equation}\label{U transition probs 2}
\begin{split}
|\langle e|U|f\rangle|^2&= 
\begin{cases} \frac{1}{2} & c=R \text{ and } u=v-1\\
\frac{1}{2} & c=L \text{ and } u=v+1 \\
0 & \text{else} \end{cases} \\
&=\frac{1}{2} \delta_{u, v-(-1)^{\delta_{c,L}}}. 
\end{split}
\end{equation}

Recall that $|\langle e|U|f\rangle|^2$ is the $(e,f)$-entry of $P$, for each $e,f\in C\times V$. 
Thus, for each $e=(c,v)\in C\times V$,  
\begin{align*}
(Pp_{X^{(\Theta,\mc T,\rho)}_0})_e &=\sum_{f\in C\times V} p_{X^{(\Theta,\mc T,\rho)}_0}(f) |\langle e|U|f\rangle|^2 \\
&= \frac{1}{2}(p_{X^{(\Theta,\mc T,\rho)}_0}(R, v-(-1)^{\delta_{c, L}}) +p_{X^{(\Theta,\mc T,\rho)}_0}(L, v-(-1)^{\delta_{c, L}})) 
\quad\text{by \eqref{U transition probs 2}} \\
&=\frac{1}{2}(\frac{3+2\sqrt 2}{N(4+2\sqrt 2)}+ \frac{1}{N(4+2\sqrt 2)}) =\frac{1}{2N}.\quad\text{by \eqref{equation0}} 
\end{align*}

Therefore $Pp_{X^{(\Theta,\mc T,\rho)}_0}\ne p_{X^{(\Theta,\mc T,\rho)}_0}$ and thus $\X^{(\Theta,\mc T,\rho)}$ is not stationary. 
Continuing to find the dynamical SZ entropy, we see that $Pp_{X^{(\Theta,\mc T,\rho)}_0}$ is the uniform distribution, $\mu$, on $C\times V$, which is invariant with respect to $P$. 
Thus Corollary~\ref{markov entropy rate 3} and Lemma~\ref{lvn gives 0 meas entropy} imply that 
\begin{align*}
h^{SZ}_{\text{dyn}}(\Theta, \mc T,\rho, \mc A)=\sum_{f\in C\times V} \mu_f \sum_{e\in C\times V} \eta(|\langle e|U|f\rangle|^2) 
= \sum_{f\in C\times V} \frac{1}{2N} 2\eta(\frac{1}{2})=\ln 2,
\end{align*}

\noindent which is equal to $h^{SZ}_{\text{dyn}}(\Theta, \mc T,\rho)$ because $\widehat{\mc A}$ is a generating partition for $(\Omega^*, \mc P(\Omega)^*, \mu^{(\Theta,\mc T, \rho)}, s)$ by Corollary~\ref{KS theorem 2}. 
\end{proof}

As the UQRW is a quantum analogue of the classical random walk, it is natural to consider measurements of the position space only. 
There are two options for how to go about this. 
One option is to take the phase space to be $(C\times V, \mc P(C\times V))$, the coherent states instrument $\mc T$ to be given by the family $(P_e)_{e\in C\times V}$, defined in Equation~\eqref{cs projections}, and calculate the dynamical SZ entropy with respect to the partition 
\begin{equation}\label{partition of V}
\mc C_V=\{C_v\}_{v\in V},\text{ where }C_v:= \{|R,v\rangle, |L,v\rangle\},\text{ for each }v\in V. 
\end{equation} 

\noindent On the other hand we could take the phase space to be $(V,\mc P(V))$, define the projections 
\begin{equation}\label{projections of V}
P_v=\1_{H_C}\ten |v\rangle\langle v|,\text{ for each }v\in V, 
\end{equation} 

\noindent and calculate the dynamical SZ entropy of the L\"uders-von Neumann instrument $\mc V$, governed by the family $(P_v)_{v\in V}$, with respect to the atomic partition of $V$. 
We will calculate the entropies for both these scenarios (with the same initial state) on the Hadamard walk, $\Theta$, and its square, $\Theta^2$. 
We will see that the two interpretations do not yield the same entropy. 
This is further evidence to the sensitivity of a closed quantum system to measurement. 
Furthermore, Theorem~\ref{coarser instrument} provides a concrete example illustrating the fact that dynamical SZ entropy is not linear in time. 

\begin{Prop}\label{Coarser Partition}
	Let $\Theta$ be the the Hadamard walk on $V$ with $|V|=N\ge 3$. 
	Let $\mc T$ be the coherent states instrument given by the family of orthogonal projections $(P_e)_{e\in C\times V}$ given in Equation~\eqref{cs projections}, $\rho=\frac{\1_H}{2N}$ and $\mc C_V$ the partition given in Equation~\eqref{partition of V}. 
	Then $h^{SZ}_{\text{dyn}}(\Theta, \mc T, \rho, \mc C_V)=\ln 2$ and $h^{SZ}_{\text{dyn}}(\Theta^2, \mc T, \rho, \mc C_V)=\frac{3}{2}\ln 2$.  
\end{Prop}

\begin{proof}
Notice that $p_{X^{(\Theta,\mc T,\rho)}_0}(c,v)=\langle c,v|\rho|c, v\rangle= \frac{1}{2N}$ for all $(c,v)\in C\times V$ and recall that the transition matrix $P$, which governs $\Theta$ with respect to the coherent states instrument $\mc T$, has $(e,f)$-entry given by Equation~\eqref{U transition probs 2}, for every $e,f\in C\times V$. In the following, it will be more convenient to rewrite Equation~\eqref{U transition probs 2} viewing $f$ as the fixed index. 
In this manner, for each $f=(c,v)\in C\times V$, we have  
\begin{equation}\label{U transitions}
U|f\rangle = \frac{1}{\sqrt{2}} (|R,v+1\rangle + (-1)^{\delta_{c, L}} |L,v-1\rangle), 
\end{equation} 

\noindent and hence 
\begin{equation}\label{U transition probs}
|\langle e|U|f\rangle|^2= 
\frac{1}{2} (\delta_{e,(R,v+1)}+\delta_{e,(L,v-1)}). 
\end{equation}

\noindent Also, we have 
\begin{align*}
\mu^{(\Theta, \mc T,\rho)}(C\left(\begin{smallmatrix} C_{v_0} & \cdots & 
C_{v_n} \\ 0 & \cdots & n\end{smallmatrix}\right))&= 
\sum_{\substack{c_k\in \{R,L\}\\ 0\le k\le n}} 
\langle c_0,v_0|\rho|c_0, v_0\rangle\prod_{k=1}^n |\langle c_k, v_k|U|c_{k-1}, v_{k-1} \rangle|^2 
\quad\text{by \eqref{cs probabilities 2}} \\ 
&= \sum_{\substack{c_k\in \{R,L\}\\ 0\le k\le n}} \frac{1}{2N}\prod_{k=1}^n \left(\frac{\delta_{v_k,v_{k-1}+1}\delta_{c_k, R}+ \delta_{v_k,v_{k-1}-1}\delta_{c_k, L}}{2}\right) 
\quad\text{by \eqref{U transition probs}} \\ 
&= \frac{1}{N} \prod_{k=1}^n \left(\frac{\delta_{v_k,v_{k-1}+1} + \delta_{v_k,v_{k-1}-1}}{2}\right),
\end{align*} 

\noindent for all $v_0,\ldots, v_n\in V$, which are exactly the probabilities $p_{\X}(v_0,\ldots,v_n)$ of a stationary Markov chain $\X$ which is governed by the transition matrix, $Q$, for the unbiased random walk on the $N$-cycle $V$.  
Therefore $h^{SZ}_{\text{dyn}}(\Theta, \mc T, \rho, \mc C_V)=H(Q)=\ln 2$, where the second equality follows from Proposition~\ref{nonzero entropy finite system}. 


Next we show that $h^{SZ}_{\text{dyn}}(\Theta^2, \mc T, \rho, \mc C_V)=\frac{3}{2}\ln 2$. 
For all $f=(c,v)\in C\times V$, we have 
\begin{equation}\label{U2 transitions}
\begin{split}
U^2|f\rangle &=
\frac{1}{\sqrt 2}U(|R,v+1\rangle + (-1)^{\delta_{c, L}} |L,v-1\rangle)
\quad\text{by \eqref{U transitions}} \\
&= \frac{1}{2}((-1)^{\delta_{R,c}}|L,v-2\rangle + (-1)^{\delta_{L,c}}|R,v\rangle + |L,v\rangle + |R,v+2\rangle)
\end{split}
\end{equation}

\noindent and hence 
\begin{equation}\label{U2 transition probs}
|\langle e|U^2|f \rangle|^2= \begin{cases} 
\frac{1}{4} & e=(R,v)\text{ or }(L,v) \\ 
\frac{1}{4} & e=(R,v+2) \\
\frac{1}{4} & e=(L,v-2) \\
0 & \text{else} \end{cases}.
\end{equation} 

\noindent Notice that $|\langle e|U^2|f \rangle|^2$ in Equation~\eqref{U2 transition probs} does not depend on the coin space component of $f$.  
Therefore  
\begin{align*}
&\mu^{(\Theta^2, \mc T,\rho)}(C\left(\begin{smallmatrix} C_{v_0} & \cdots & 
C_{v_n} \\ 0 & \cdots & n\end{smallmatrix}\right)) \\
&=\sum_{\substack{c_k\in \{R,L\}\\ 0\le k\le n}} 
\langle c_0,v_0|\rho|c_0, v_0\rangle\prod_{k=1}^n |\langle c_k, v_k|U^2|c_{k-1}, v_{k-1} \rangle|^2 
\quad\text{by \eqref{cs probabilities 2}} \\
&=\sum_{\substack{c_k\in \{R,L\}\\ 0\le k\le n}} 
\frac{1}{2N}\prod_{k=1}^n \left(\frac{ 
\delta_{c_k, L}\delta_{v_k, v_{k-1}-2}+ \delta_{c_k, L}\delta_{v_k, v_{k-1}}+
\delta_{c_k, R}\delta_{v_k, v_{k-1}}+\delta_{c_k, R}\delta_{v_k, v_{k-1}+2}}{4}\right) 
\quad\text{by \eqref{U2 transition probs}} \\
&= \frac{1}{N}\prod_{k=1}^n \left(\frac{1}{4} \delta_{v_k, v_{k-1}-2}+
\frac{1}{2} \delta_{v_k, v_{k-1}}+
\frac{1}{4} \delta_{v_k, v_{k-1}+2}\right),
\end{align*}

\noindent for all $v_0,\ldots, v_n\in V$, which are exactly the probabilities $p_{\Y}(v_0,\ldots,v_n)$ of a stationary Markov chain $\Y$ which is governed by the transition matrix $Q^2$. 
Therefore $h^{SZ}_{\text{dyn}}(\Theta^2 , \mc T, \rho, \mc C_V)=H(Q^2)=\frac{3}{2}\ln 2$, where the second equality follows from Proposition~\ref{nonzero entropy finite system}. 
\end{proof}

\begin{theorem}\label{coarser instrument}
	Let $\Theta$ be the Hadamard walk on $V$ with $|V|=N\ge 3$. 
	Let $\mc V$ be the L\"uders-von Neumann instrument given by the family of orthogonal rank-2 projections $(P_v)_{v\in V}$ defined in Equation~\eqref{projections of V} and $\rho=\frac{\1_H}{2N}$. 
	Then $h^{SZ}_{\text{dyn}}(\Theta, \mc V, \rho)=\ln 2$ and $h^{SZ}_{\text{dyn}}(\Theta^2, \mc V, \rho)=\frac{4}{3}\ln 2$.  
\end{theorem}

\begin{proof} 
Notice that from Equation~\eqref{general entropy definition}, for each $m\in\N$, $n\in\N_0$ and $v_0,\ldots, v_n\in V$, we have   
\begin{equation}\label{coarser instrument probabilities}
\begin{split}
\mu^{(\Theta^m,\mc V,\rho)}(C\left(\begin{smallmatrix} v_0 & \cdots & 
v_n \\ 0 & \cdots & n\end{smallmatrix}\right))&:=
\tau(\mc V(v_n)\circ \Theta^m\circ\cdots\circ \Theta^m\circ \mc V(v_0) \rho) \\
&= \tr( P_{v_n}U^m\cdots U^mP_{v_0}\rho P_{v_0}(U^m)^*\cdots (U^m)^* P_{v_n}).
\end{split}
\end{equation}

\noindent Also, notice that $\rho=\frac{1}{2N}\sum_{v\in V} (|R,v\rangle\langle R,v|+|L,v\rangle\langle L,v|)$ and so, for each $m\in \N$, Equation~\eqref{coarser instrument probabilities} becomes 
\begin{align}\nonumber
\mu^{(\Theta^m,\mc V,\rho)}(C\left(\begin{smallmatrix} v_0 & \cdots & 
v_n \\ 0 & \cdots & n\end{smallmatrix}\right))
&= \sum_{c\in \{R,L\}} \frac{1}{2N}\tr( P_{v_n}U^m\cdots U^mP_{v_0}|c,v_0\rangle\langle c,v_0| P_{v_0}(U^m)^*\cdots (U^m)^* P_{v_n}) \\ \label{coarser instrument probabilities 2}
&= \sum_{c,d \in \{R,L\}} \frac{1}{2N} |\langle d,v_n| U^m P_{v_{n-1}}\cdots P_{v_1}U^m|c,v_0\rangle|^2. 
\end{align}
 
Let $\mc A$ be the atomic partition of $V$. 
We first show that $h^{SZ}_{\text{dyn}}(\Theta, \mc V, \rho, \mc A)=\ln 2$. 
Notice that for $(c,v)\in C\times V$, $UP_v |c,v\rangle = U|c,v\rangle$ and is given by Equation~\eqref{U transitions}.  
Thus, by direct calculation, we have that 
\begin{align*}
\mu^{(\Theta,\mc V,\rho)}&(C\left(\begin{smallmatrix} v_0 & \cdots & 
v_n \\ 0 & \cdots & n\end{smallmatrix}\right))\\
&=
\sum_{c_0,c_n \in \{R,L\}} \frac{1}{2N} |\langle c_n,v_n| U P_{v_{n-1}}\cdots P_{v_1}U|c_0,v_0\rangle|^2 
\quad\text{by \eqref{coarser instrument probabilities 2}} \\
&= \frac{1}{2N}\sum_{c_0, c_n \in \{R,L\}} |\langle c_n,v_n| U P_{v_{n-1}}\cdots P_{v_1} (\frac{1}{\sqrt 2} (|R,v_0+1\rangle +(-1)^{\delta_{c_0,L}} |L,v_0-1\rangle))|^2 \\
&=  \frac{1}{2N}\sum_{c_0, c_1, c_n \in \{R,L\}} \frac{1}{2}(\delta_{v_1,v_{0}+1}\delta_{c_1, R} + \delta_{v_1,v_0-1}\delta_{c_1, L}) |\langle c_n,v_n| U P_{v_{n-1}}\cdots P_{v_2}U|c_1,v_1\rangle|^2 \\
&= \cdots \\
&= \sum_{\substack{c_k\in \{R,L\}\\ 0\le k\le n}} \frac{1}{2N}\prod_{k=1}^n \left(\frac{\delta_{v_k,v_{k-1}+1}\delta_{c_k, R}+ \delta_{v_k,v_{k-1}-1}\delta_{c_k, L}}{2}\right) \\ 
&= \frac{1}{N} \prod_{k=1}^n \left(\frac{\delta_{v_k,v_{k-1}+1} + \delta_{v_k,v_{k-1}-1}}{2}\right),
\end{align*}

\noindent for all $v_0,\ldots, v_n\in V$, which are exactly the probabilities $p_{\X}(v_0,\ldots,v_n)$ of a stationary Markov chain $\X$ which is governed by the transition matrix, $Q$, for the unbiased random walk on the $N$-cycle $V$. 
Therefore $h^{SZ}_{\text{dyn}}(\Theta, \mc V, \rho, \mc A)=H(Q)=\ln 2$, where the second equality follows from Proposition~\ref{nonzero entropy finite system}. 
Moreover, since $\widehat{\mc A}$ is a generating partition for $(V^*,\mc P(V)^*,\mu^{(\Theta,\mc V,\rho)},s)$, we have that  $h^{SZ}_{\text{dyn}}(\Theta, \mc V, \rho)=\ln 2$ by Corollary~\ref{KS theorem 2}.


Next we show that $h^{SZ}_{\text{dyn}}(\Theta^2, \mc V, \rho, \mc A)=\frac{4}{3}\ln 2$ using path counting techniques. 
To that end, for each $n\in\N_0$ and $(n+1)$-tuple $v= (v_0, \ldots, v_n)\in V^{n+1}$, we set $$l_{v}:=|\{ k : k<n \text{ such that } v_{k}=v_{k+1}=\cdots=v_n\}|.$$ 
Then, we define the sets 
$$L^n_{\text{c}}:=\{\bar v\in V^{n+1}: \bar v=(v,\ldots, v)\text{ for some }v\in V\},$$
$$L^n_{\text{e}}:=\{v\in V^{n+1}: l_{v} \text{ is even}\}\backslash L^n_{\text{c}}\text{ and } 
L^n_{\text{o}}:=\{v\in V^{n+1}: l_{v} \text{ is odd}\}\backslash L^n_{\text{c}}.$$
 
\noindent For each $n\in \N_0$ and $v=(v_0,\ldots, v_n)\in V^{n+1}$, we will identify $v$ with the cylinder set $C\left(\begin{smallmatrix} v_0 & \cdots & 
v_n \\ 0 & \cdots & n\end{smallmatrix}\right)$ and consider $\mc L^n:=\{L^n_{\text{c}}, L^n_{\text{e}}, L^n_{\text{o}}\}$ as a partition of $V^*$.


With this correspondence, we will show that the conditional probabilities, $p_{\X^{(\Theta^2,\mc V,\rho)}}(v_{n+1}|v)$, for $v_{n+1}$ given $v=(v_0,\ldots, v_n)$ are dependent upon which set $L\in\mc L^n$ that $v$ belongs to. 
First, we will determine the change of coin state that occurs after measuring the walker at the same site a number of times in a row. 
We claim that the resulting coin state, after $n$ measurements at a site $v$, depends only on the initial coin state $c\in\{R,L\}$ and the congruence class of $n$ modulo 4. 
Specifically, for all $n\in\N_0$, $\bar v=(v,\ldots, v)\in V^{n+1}$ and $c\in\{R,L\}$, we claim the following: 
if $n\equiv 0 \mod 4$, then 
\begin{equation}\label{equation 11}
\underbrace{P_{v}U^2\cdots P_{v} U^2}_{n\text{ times}}  |c,v\rangle= a|c,v\rangle\text{ for some $a\in\C$ with $|a|=\frac{1}{2^{\lfloor\frac{n+1}{2}\rfloor}}$,}
\end{equation}

\noindent if $n\equiv 1\mod 4$, then 
\begin{equation}\label{equation 12}
P_{v}U^2\cdots P_{v} U^2 |c,v\rangle
=a|L+ (-1)^{\delta_{c,L}} R, v\rangle, \text{ for some $a\in \C$ with $|a|=\frac{1}{2^{\lfloor\frac{n+1}{2}\rfloor}}$,}
\end{equation}

\noindent if $n\equiv 2\mod 4$, then 
\begin{equation}\label{equation 11'}
P_{v}U^2\cdots P_{v} U^2 |c,v\rangle= a|c^\perp,v\rangle\text{ for some $a\in\C$ with $|a|=\frac{1}{2^{\lfloor\frac{n+1}{2}\rfloor}}$,}
\end{equation}

\noindent where we set $R^\perp=L$ and $L^\perp=R$, and, if $n\equiv 3\mod 4$, then 
\begin{equation}\label{equation 12'}
P_{v}U^2\cdots P_{v} U^2 |c,v\rangle
=a|L- (-1)^{\delta_{c,L}} R, v\rangle, \text{ for some $a\in \C$ with $|a|=\frac{1}{2^{\lfloor\frac{n+1}{2}\rfloor}}$,}
\end{equation}
\noindent where we used the abbreviation $|L\pm R,v\rangle:=|L,v\rangle \pm |R,v\rangle$. 
We will prove the claims by induction on $n$.


The base case, $n=0$, is trivial. 
For the inductive step we will handle the different congruence classes of $n$ separately. 
To this end, let $m\in\N$ with $m\ge 1$ and suppose that for all $n<m$ Equations~\eqref{equation 11}-\eqref{equation 12'} hold for all $\bar v\in V^{n+1}$ and their respective values of $n$. 
Fix $\bar v\in V^{m+1}$ and $c\in\{R, L\}$. 
If $m\equiv 1\mod 4$, then 
\begin{equation*}
\begin{split}
P_{v}U^2\cdots P_v U^2  |c,v\rangle&= P_{v}U^2a|c,v\rangle \text{ for some $a\in\C$ with $|a|=\frac{1}{2^{\lfloor\frac{m}{2}\rfloor}}$ by \eqref{equation 11}}\\
&=\frac{a}{2}|L+ (-1)^{\delta_{c,L}} R, v\rangle\text{ by \eqref{U2 transitions},}
\end{split}
\end{equation*}

\noindent and Equation~\eqref{equation 12} is satisfied since $\frac{1}{2\cdot 2^{\lfloor\frac{m}{2}\rfloor}}=\frac{1}{2^{\lfloor\frac{m+1}{2}\rfloor}}$. 
If $m\equiv 2\mod 4$, then 
\begin{align*}
P_{v}U^2\cdots P_{v} U^2  |c,v\rangle
&= P_{v}U^2a|L+ (-1)^{\delta_{c,L}} R,v\rangle \text{ for some $a\in\C$ with $|a|=\frac{1}{2^{\lfloor\frac{m}{2}\rfloor}}$ by \eqref{equation 12}}\\
&=(-1)^{\delta_{c,L}}a|c^\perp, v\rangle,
\end{align*}

\noindent where the second equality holds because  
\begin{equation}\label{equation 3}
\begin{split}
&U^2 |R+L, v\rangle  = \sqrt{2} U|R,v+1\rangle 
= |L, v\rangle + 
|R, v+2\rangle\text{ and } \\
&U^2 |L-R, v \rangle = -\sqrt{2} U |L,v-1\rangle 
= |L, v-2 \rangle - |R,v\rangle, 
\end{split}
\end{equation}

\noindent for all $v\in V$. 
Thus Equation~\eqref{equation 11'} is satisfied since $\frac{1}{ 2^{\lfloor\frac{m}{2}\rfloor}}=\frac{1}{2^{\lfloor\frac{m+1}{2}\rfloor}}$. 
If $m\equiv 3\mod 4$, then 
\begin{equation*}
\begin{split}
P_{v}U^2\cdots P_{v} U^2 |c,v\rangle&= P_{v}U^2a|c^\perp,v\rangle \text{ for some $a\in\C$ with $|a|=\frac{1}{2^{\lfloor\frac{m}{2}\rfloor}}$ by \eqref{equation 11'}}\\
&=\frac{a}{2}|L- (-1)^{\delta_{c,L}} R, v\rangle\text{ by \eqref{U2 transitions},}
\end{split}
\end{equation*}

\noindent where we used the fact that $(-1)^{\delta_{c^\perp,L}}=-(-1)^{\delta_{c,L}}$ in the second equality. 
Hence Equation~\eqref{equation 12'} is satisfied since $\frac{1}{2\cdot 2^{\lfloor\frac{m}{2}\rfloor}}=\frac{1}{2^{\lfloor\frac{m+1}{2}\rfloor}}$. 
If $m\equiv 0\mod 4$, then 
\begin{align*}
P_{v}U^2\cdots P_{v} U^2 |c,v\rangle
&= P_{v}U^2a|L- (-1)^{\delta_{c,L}} R,v\rangle \text{ for some $a\in\C$ with $|a|=\frac{1}{2^{\lfloor\frac{m}{2}\rfloor}}$ by \eqref{equation 12'}}\\
&=-(-1)^{\delta_{c,L}}a|c, v\rangle\text{ by \eqref{equation 3}}
\end{align*}

\noindent and hence Equation~\eqref{equation 11} is satisfied since $\frac{1}{ 2^{\lfloor\frac{m}{2}\rfloor}}=\frac{1}{2^{\lfloor\frac{m+1}{2}\rfloor}}$. 
Therefore the induction is complete and the claims are verified.


Next we claim that for all $v=(v_0,\ldots, v_n)\in V^{n+1}\bs L^n_{\text{c}}$ such that $p_{\X^{(\Theta^2,\mc V,\rho)}}(v)\ne 0$ there exists some $\psi\in H_C$ such that for all $v_{n+1}\in V$ the conditional pmf of $\X^{(\Theta^2,\mc V,\rho)}$ is given by 
\begin{equation}\label{equation 15'}
p_{\X^{(\Theta^2,\mc V,\rho)}}(v_{n+1}|v_0,\ldots,v_n)
= \frac{\sum_{d\in\{R,L\}}|\langle d,v_{n+1}| U^2|\psi,v_n\rangle |^2}{\|\psi\|^2}. 
\end{equation}

Indeed, for all $v\in V^{n+1}$ with $p_{\X^{(\Theta^2,\mc V,\rho)}}(v)\ne 0$, we have 
\begin{align}\nonumber 
p_{\X^{(\Theta^2,\mc V,\rho)}}(v_{n+1}|v_0,\ldots,v_n)&=
\frac{\mu^{(\Theta^2,\mc V,\rho)}(C\left(\begin{smallmatrix} v_0 & \cdots & 
	v_{n+1} \\ 0 & \cdots & n+1\end{smallmatrix}\right))}{\mu^{(\Theta^2,\mc V,\rho)}(C\left(\begin{smallmatrix} v_0 & \cdots & 
	v_n \\ 0 & \cdots & n\end{smallmatrix}\right))} \\ \nonumber
&=\frac{
	\sum_{c,d \in \{R,L\}}  |\langle d,v_{n+1}| U^2 P_{v_n}\cdots P_{v_1}U^2|c,v_0\rangle|^2}{\sum_{c,d \in \{R,L\}}  |\langle d,v_{n}| U^2 P_{v_{n-1}}\cdots P_{v_1}U^2|c,v_0\rangle|^2}
\text{ by \eqref{coarser instrument probabilities 2}}  
\\ \label{eqn 15}
&=  \sum_{c\in\{R,L\}} q_c \frac{
	\sum_{d \in \{R,L\}}  |\langle d,v_{n+1}| U^2 P_{v_n}\cdots P_{v_1}U^2|c,v_0\rangle|^2}{\sum_{d \in \{R,L\}}  |\langle d,v_{n}| U^2 P_{v_{n-1}}\cdots P_{v_1}U^2|c,v_0\rangle|^2},
\end{align}

\noindent where, for each $c\in\{R,L\}$, we set 
\begin{equation}\label{ac probs}
q_c:= \frac{\sum_{d \in \{R,L\}}  |\langle d,v_n| U^2 P_{v_{n-1}}\cdots P_{v_1}U^2|c,v_0\rangle|^2}{\sum_{c', d \in \{R,L\}}  |\langle d,v_n| U^2 P_{v_{n-1}}\cdots P_{v_1}U^2|c',v_0\rangle|^2}.
\end{equation} 

\noindent Notice that if $q_c=0$  in Equation~\eqref{ac probs} then the denominator on the right hand side of Equation~\eqref{eqn 15} is also 0. In this case, we will use the convention that their product is defined and equal to 0. 
For each $c\in\{R,L\}$, we define $\psi_c\in H_C$ to be the unique element satisfying the equation 
\begin{equation*}\label{n coin state}
P_{v_n}U^2\cdots P_{v_1}U^2|c,v_0\rangle=|\psi_c,v_n\rangle. 
\end{equation*}

\noindent Then Equation~\eqref{eqn 15} simplifies to 
\begin{equation}\label{equation 15}
p_{\X^{(\Theta^2,\mc V,\rho)}}(v_{n+1}|v_0,\ldots,v_n)
= \sum_{c\in\{R,L\}} q_c \frac{\sum_{d\in\{R,L\}}|\langle d,v_{n+1}| U^2|\psi_c,v_n\rangle |^2}{\|\psi_c\|^2}, 
\end{equation}

\noindent where equality in the denominator follows by Parseval's identity. 

Notice that Equation~\eqref{U2 transitions} gives  
\begin{equation}\label{equation 13}
\begin{split}
&\Ran(P_{v+2}U^2 P_v)=\text{span} \{|R,v+2\rangle\}\text{ and }\\ &\Ran(P_{v-2}U^2 P_v)=\text{span} \{|L,v-2\rangle\},\text{ for all $v\in V$.} 
\end{split} 
\end{equation}

\noindent Moreover, Equation~\eqref{equation 13} implies that, for any operator $A\in B(H)$,
\begin{equation*}
\begin{split}
&\Ran(P_{v+2}U^2 P_vA)\subseteq\text{span} \{|R,v+2\rangle\}\text{ and }\\ &\Ran(P_{v-2}U^2 P_vA)\subseteq\text{span} \{|L,v-2\rangle\},\text{ for all $v\in V$.} 
\end{split} 
\end{equation*}

\noindent Hence, if $v\in V^{n+1}\bs L^n_{\text{c}}$, then for each $c\in\{R,L\}$ we have  
\begin{equation}\label{equation 14}
P_{v_{n-l_v}}U^2\cdots P_{v_1} U^2 |c,v_0\rangle=a_c|d,v_{n-l_v}\rangle, \text{ for some }a_c\in \C,
\end{equation} 
where $d=R$ whenever $v_{n-l_v}=v_{n-l_v-1}+2$ and $d=L$ when $v_{n-l_v}=v_{n-l_v-1}-2$. 
Thus $d$ does not depend on the initial coin state $c$. 
For each $d\in\{R,L\}$, we also have that 
\begin{equation}\label{equation 16}
P_{v_n}U^2\cdots P_{v_{n-l_v+1}} U^2  |d,v_{n-l_v}\rangle
=a|\psi, v_n\rangle, \text{ for some $a\in \C$ with $|a|=\frac{1}{2^{\lfloor\frac{l_v+1}{2}\rfloor}}$,} 
\end{equation}

\noindent where $\psi$ is the coin state given by Equations~\eqref{equation 11}-\eqref{equation 12'} depending on the congruence class of $l_v$ modulo 4 and we used the fact that $v_{n-l_v}=v_{n-l_v+1}=\cdots=v_n$ by definition of $l_v$. 
We combine Equations~\eqref{equation 14} and \eqref{equation 16} to get that, for each $c\in\{R,L\}$,  
\begin{equation}\label{equation 17}
P_{v_{n}}U^2\cdots P_{v_1} U^2 |c,v_0\rangle=a_c'|\psi, v_n\rangle, 
\end{equation}

\noindent where $a_c'=a_c\cdot a$ with $a_c$ and $a$ coming from Equations~\eqref{equation 14} and \eqref{equation 16}, respectively. 
Since $q_R+q_L=1$ and both $\psi_R$ and $\psi_L$ in Equation~\eqref{equation 15} are equal to $\psi$ which appears in Equation~\eqref{equation 17}, we see that Equation~\eqref{equation 15} simplifies to Equation~\eqref{equation 15'} as claimed. 
%


Next we claim that, for $n\in\N_0$, if $v=(v_0,\ldots, v_n)\in L^n_{\text{o}}$ and $p_{\X^{(\Theta^2,\mc V,\rho)}}(v)\ne 0$, then 
\begin{equation}\label{equation 6}
p_{\X^{(\Theta^2,\mc V,\rho)}}(v_{n+1}|v_0,\ldots,v_n)
= \begin{cases}
\frac{1}{2} & \text{if }v_{n+1}=v_n \\
\frac{1}{2} & \text{if }v_{n+1}\text{ is exactly one of }v_n\pm 2 \\
0 & \text{else}
\end{cases},
\end{equation}

\noindent where the exactly one value of $v_{n+1}\in\{v_n-2,v_n+2\}$ with nonzero conditional probability depends on the given sequence $(v_0,\ldots,v_n)$ in the following manner:
\begin{tabular}{l r l}
	$\cdott$ If & & \\
	& (i) & $v_{n-l_v}=v_{n-l_v-1}+2 \text{ and } l_v= 1\text{ mod } 4$, or \\
	& (ii) & $v_{n-l_v}=v_{n-l_v-1}-2 \text{ and } l_v= 3\text{ mod } 4$, \\ \multicolumn{3}{l}{\hspace{0.55cm}then $v_{n+1}=v_n+2$.}
\end{tabular} \\
\begin{tabular}{l r l}
	$\cdott$ If & & \\
	& (iii) & $v_{n-l_v}=v_{n-l_v-1}-2 \text{ and } l_v= 1\text{ mod } 4$, or \\
	& (iv) & $v_{n-l_v}=v_{n-l_v-1}+2 \text{ and } l_v= 3\text{ mod } 4$, \\ \multicolumn{3}{l}{\hspace{0.55cm}then $v_{n+1}=v_n-2$.}
\end{tabular}

\noindent In addition we claim that, for $n\in\N_0$ and $v=(v_0,\ldots, v_n)\in V^{n+1}$, if $v\in L^n_{\text{e}}\cup L^n_{\text{c}}$ and $p_{\X^{(\Theta^2,\mc V,\rho)}}(v)\ne 0$, then
\begin{equation}\label{equation 5}
p_{\X^{(\Theta^2,\mc V,\rho)}}(v_{n+1}|v_0,\ldots,v_n)
=\begin{cases}
\frac{1}{2} & \text{if }v_{n+1}=v_n \\
\frac{1}{4} & \text{if }v_{n+1}=v_n+ 2 \\
\frac{1}{4} & \text{if }v_{n+1}=v_n-2 \\
0 & \text{else}
\end{cases}
\end{equation}


In order to see Equation~\eqref{equation 6}, let $v\in L^n_{\text{o}}$ with $p_{\X^{(\Theta^2,\mc V,\rho)}}(v)\ne 0$ and suppose $v$ satisfies the conditions for Case~(i); i.e. $v_{n-l_v}=v_{n-l_v-1}+2$ and $l_v\equiv 1\mod 4$. 
Since $v_{n-l_v}=v_{n-l_v-1}+2$, the coin state, $d$, on the right hand side of Equation~\eqref{equation 14} is $d=R$. 
Using this, the fact that $l_v\equiv 1\mod 4$ and Equation~\eqref{equation 12}, we see that the coin state, $\psi$, on the right hand sides of Equations~\eqref{equation 16} and \eqref{equation 17} is given by $\psi=R+L$. 
Plugging into Equation~\eqref{equation 15'} and using \eqref{equation 3}, we have 
\begin{equation*}
p_{\X^{(\Theta^2,\mc V,\rho)}}(v_{n+1}|v_0,\ldots,v_n)
= \begin{cases}
\frac{1}{2} & \text{if }v_{n+1}=v_n \\
\frac{1}{2} & \text{if }v_{n+1}=v_n+ 2 \\
0 & \text{else}
\end{cases}
\end{equation*}

\noindent in this case. The other three cases can be done similarly and thus we obtain that Equation~\eqref{equation 6} is satisfied for all $v\in L^n_{\text{o}}$.


Next, for the proof of Equation~\eqref{equation 5}, let $v\in L^n_{\text{e}}$ with $p_{\X^{(\Theta^2,\mc V,\rho)}}(v)\ne 0$. 
By Equations~\eqref{equation 11} and \eqref{equation 11'}, we can see that the coin state $\psi$ in Equation~\eqref{equation 17} is given by $\psi=c$, for some $c\in\{R,L\}$. 
Plugging this value of $\psi$ into Equation~\eqref{equation 15'} and using 
\eqref{U2 transition probs} we can see that the conditional pmf, $p_{\X^{(\Theta^2,\mc V,\rho)}}(v_{n+1}|v_0,\ldots,v_n)$, of $\X^{(\Theta^2,\mc V,\rho)}$ is given by Equation~\eqref{equation 5}, for all $v\in L^n_{\text{e}}$.


It remains only to show that Equation~\eqref{equation 5} is valid for all $\bar v=(v,\ldots, v)\in L^n_{\text{c}}$. 
Since the modulus of $a$ in Equations~\eqref{equation 11}-\eqref{equation 12'} is independent of $c\in\{R,L\}$, we have $q_c=\frac{1}{2}$ in Equation~\eqref{ac probs} for both values of $c$. 
Note that by Equations~\eqref{equation 11}-\eqref{equation 12'} we have that the vector $\psi_c\in H_C$ which appears in Equation~\eqref{equation 15} is given by 
$$ \psi_c= \begin{cases}
c & \text{if }n\equiv 0\text{ mod }4 \\
L+(-1)^{\delta_{c,L}}R & \text{if }n\equiv 1\text{ mod }4 \\
c^\perp & \text{if }n\equiv 2\text{ mod }4 \\
L-(-1)^{\delta_{c,L}}R & \text{if }n\equiv 3\text{ mod }4
\end{cases}$$
Thus if $n$ is even then by Equations~\eqref{equation 15} and \eqref{U2 transition probs} we obtain immediately Equation~\eqref{equation 5}. 
If $n$ is odd we examine the cases $n\equiv 1\mod 4$ and $n\equiv 3\mod 4$ separately. 
If $n\equiv 1\mod 4$ then $\psi_R=L+R$, $\psi_L=L-R$ and Equations~\eqref{equation 15} and \eqref{equation 3} give Equation~\eqref{equation 5}. 
The case of $n\equiv 3\mod 4$ can be verified similarly.


We are now set to show $h^{SZ}_{\text{dyn}}(\Theta^2, \mc V, \rho, \mc A)=\frac{4}{3}\ln 2$.  
By direct calculation, we have that 
\begin{align}\nonumber
H_{\mu^{(\Theta^2,\mc V,\rho)}}&(X^{(\Theta^2,\mc V,\rho)}_{n+1}|(X^{(\Theta^2,\mc V,\rho)}_{0},\ldots, X^{(\Theta^2,\mc V,\rho)}_{n})) \\ \nonumber
&= \sum_{\substack{v_k\in V\\ 0\le k\le n}} p_{\X^{(\Theta^2,\mc V,\rho)}}(v_0,\ldots,v_n)
\sum_{v_{n+1}\in V} \eta(p_{\X^{(\Theta^2,\mc V,\rho)}}(v_{n+1}|v_0,\ldots,v_n))
\quad\text{by \eqref{conditional entropy}} \\
\nonumber 
&= \mu^{(\Theta^2,\mc V,\rho)}(L^n_{\text e}\cup L^n_{\text{c}})(2\eta(\frac{1}{4})+\eta(\frac{1}{2})) +
\mu^{(\Theta^2,\mc V,\rho)}(L^n_{\text o})(2\eta(\frac{1}{2})) 
\text{ by \eqref{equation 6} and \eqref{equation 5}} \\ \label{equation 8}
&= \mu^{(\Theta^2,\mc V,\rho)}(L^n_{\text e}\cup L^n_{\text{c}})\frac{3}{2}\ln 2 + \mu^{(\Theta^2,\mc V,\rho)}(L^n_{\text o})\ln 2,
\end{align}

\noindent where $$\mu^{(\Theta^2,\mc V,\rho)}(L^n_{x})= \sum_{(v_0,\ldots, v_n)\in L^n_{x}} p_{\X^{(\Theta^2,\mc V,\rho)}}(v_0,\ldots,v_n)\text{ for each $x\in \{\text{c}, \text{e}, \text{o}\}.$}$$ 
It remains only to solve for $$\lim_{n\ra\infty}\mu^{(\Theta^2,\mc V,\rho)}(L^n_{x})\text{ for each $x\in \{\text{c},\text{e},\text{o}\}$.}$$ 
Notice that, by definition of $L^n_{\text{c}}$, $$\mu^{(\Theta^2,\mc V,\rho)}(L^n_{\text{c}})= \sum_{v\in V} p_{\X^{(\Theta^2,\mc V,\rho)}}(\underbrace{v, \ldots,v}_{n+1\text{ times}}), \text{ for each $n\in \N_0.$}$$ 
Using Equation~\eqref{equation 5} $n$ times and Equation~\eqref{coarser instrument probabilities 2} to see that $p_{\X^{(\Theta^2,\mc V,\rho)}}(v)=\frac{1}{N}$, we obtain that $p_{\X^{(\Theta^2,\mc V,\rho)}}(\underbrace{v, \ldots,v}_{n+1\text{ times}})=\frac{1}{2^nN}$ and hence 
\begin{equation}\label{fixed probabilities}
\mu^{(\Theta^2,\mc V,\rho)}(L^n_{\text{c}}) =\frac{1}{2^n},\text{ for each $n\in\N_0$.} 
\end{equation}  

For ease of notation, set $e_n=\mu^{(\Theta^2,\mc V,\rho)}(L^n_{\text e})$, $o_n=\mu^{(\Theta^2,\mc V,\rho)}(L^n_{\text{o}})$ and $c_n=\mu^{(\Theta^2,\mc V,\rho)}(L^n_{\text{c}})$, for each $n\in \N_0$. 
Notice that $e_0=o_0=0$, $c_0=1$ and $c_n=\frac{1}{2^n}$, for all $n\in\N$, by Equation~\eqref{fixed probabilities}. 
For each $v=(v_0,\ldots, v_{n-1})\in V^n$ and $v_n\in V$ let $v\circ v_n=(v_0,\ldots, v_n)\in V^{n+1}$. 
Suppose $v\in L^{n-1}_{\text{o}}$. If $v_n=v_{n-1}$ then $l_v+1=l_{v\circ v_n}$ and if $v_n\ne v_{n-1}$ then $l_{v\circ v_n}=0$ and thus $p_{\X^{(\Theta^2,\mc V,\rho)}}(v\circ v_n\in L^n_{\text{e}}|v\in L^n_{\text{o}})=1$ and $p_{\X^{(\Theta^2,\mc V,\rho)}}(v\circ v_n\in L^n_{x}|v\in L^n_{\text{o}})=0$ for $x\in\{\text{o, c}\}$. 
Suppose $v\in L^{n-1}_{\text{e}}\cup L^{n-1}_{\text{c}}$. Then $v\circ v_n\in L^n_{\text{e}}$ exactly when $v_n\ne v_{n-1}$. Therefore Equation~\eqref{equation 5} gives $p_{\X^{(\Theta^2,\mc V,\rho)}}(v\circ v_n\in L^n_{\text{e}}|v\in L^n_{x})=\frac{1}{2}$ for $x\in\{\text{e, c}\}$. 
Therefore 
\begin{equation}\label{a' recursion}
e_n=o_{n-1} +\frac{1}{2}(e_{n-1}+c_{n-1}),\text{ for all $n\in\N$.}
\end{equation}

\noindent Equation~\eqref{equation 5} also gives that $p_{\X^{(\Theta^2,\mc V,\rho)}}(v\circ v_n\in L^n_{\text{o}}|v\in L^n_{\text{e}})=\frac{1}{2}$ and since $p_{\X^{(\Theta^2,\mc V,\rho)}}(v\circ v_n\in L^n_{\text{o}}|v\in L^n_{\text{c}})=0$, we have 
\begin{equation}\label{b recursion}
o_n= \frac{1}{2}e_{n-1},\text{ for all }n\in\N.
\end{equation}

\noindent Therefore 
\begin{equation}\label{a recursion}
e_n=\frac{1}{2}(e_{n-1}+e_{n-2}+c_{n-1}), \text{ for all }n\ge 2.
\end{equation}

\noindent We claim that the limits $e:=\lim_{n\ra\infty}e_n$ and $o:=\lim_{n\ra\infty}o_n$ both exist. It is enough, by Equation~\eqref{b recursion}, to show that the limit $e$ exists. 
To see this we show that $(e_n)_{n\in\N_0}$ is increasing and bounded. 
We show that $(e_n)_{n\in\N_0}$ is increasing by induction. Since $e_0=0$ and $e_1=\frac{1}{2}$ (by Equation~\eqref{a' recursion}), the base case is done. Next, fix $n\in\N$ with $n\ge 2$, suppose that $e_{m-1}< e_m$ for all $m\in\{1,\ldots, n-1\}$. 
Then, by Equation~\eqref{a recursion}, it is enough to show that $e_{n-2}+c_{n-1}> e_{n-1}$. We see that, 
\begin{align*}
e_{n-1}= \frac{1}{2}(e_{n-2}+e_{n-3}+c_{n-2})< e_{n-2}+c_{n-1},
\end{align*}

\noindent where the inequality follows by the inductive hypothesis and the fact that $\frac{1}{2}c_{n-2}=c_{n-1}$. 
Therefore $(e_n)_{n\in\N_0}$ is increasing and trivially bounded by 1 and both the limits $e$ and $o$ exist. 
Furthermore, $1=e_n+o_n+c_n$, for all $n\in\N_0$, because $\mc L^n$ is a partition of $V^*$ and $\lim_{n\ra\infty}c_n=0$ by Equation~\eqref{fixed probabilities}. 
Hence $1=e+o=\frac{3e}{2}$, $e=\frac{2}{3}$ and $o=\frac{1}{3}$. 
Taking the limit in Equation~\eqref{equation 8}, we see that 
\begin{equation*}
\lim_{n\ra\infty}(e_n+c_n)\frac{3}{2}\ln 2 +
o_n\ln 2= \frac{4}{3}\ln 2.
\end{equation*}

Therefore $h^{SZ}_{\text{dyn}}(\Theta^2,\mc V, \rho) =h^{SZ}_{\text{dyn}}(\Theta^2,\mc V, \rho, \mc A)= \frac{4}{3}\ln 2$ as desired. 
\end{proof}

The fact that $h^{SZ}_{\text{dyn}}(\Theta^2,\mc V, \rho, \mc A)\ne 
h^{SZ}_{\text{dyn}}(\Theta^2 , \mc T, \rho, \mc C_V)$ provides further evidence of the sensitivity of quantum systems to measurement.

%
%

\bibliographystyle{amsplain}

\bibliography{../references}

\end{document}